\documentclass[a4paper,USenglish,cleveref,pdfa]{lipics-v2021}

\pdfoutput=1
\hideLIPIcs
\nolinenumbers

\usepackage{preamble}

\title{Bandwidth Parameterized by Cluster Vertex Deletion Number%
\footnote{An extended abstract of this work was presented at the 18th International Symposium on Parameterized and Exact Computation
(IPEC 2023)~\cite{iwpec/Gima0KMV23}.}}

\titlerunning{Bandwidth Parameterized by Cluster Vertex Deletion Number}

\author{Tatsuya Gima}
{Hokkaido University, Sapporo, Japan}
{gima@ist.hokudai.ac.jp}
{https://orcid.org/0000-0003-2815-5699}
{Partially supported by JSPS KAKENHI Grant Number JP23KJ1066.}

\author{Eun Jung Kim}
{School of Computing, KAIST and Discrete Mathematics Group, Institute for Basic Science (IBS), Daejeon, South Korea / CNRS, France}
{eunjung.kim@kaist.ac.kr}
{https://orcid.org/0000-0002-6824-0516}
{Supported by ANR project ANR-18-CE40-0025-01 (ASSK), Institute for Basic Science (IBS-R029-C1) and National Research Foundation in Korea (RS-2025-00563533).}

\author{Noleen K\"{o}hler}
{University of Leeds, Leeds, UK}
{n.koehler@leeds.ac.uk}
{https://orcid.org/0000-0002-1023-6530}
{Supported by ANR project ANR-18-CE40-0025-01 (ASSK).}

\author{Nikolaos Melissinos}
{Department of Theoretical Computer Science, Faculty of Information Technology,
Czech Technical University in Prague, Czech Republic}{nikolaos.melissinos@fit.cvut.cz}{https://orcid.org/0000-0002-0864-9803}{Supported by the CTU Global postdoc fellowship program.}

\author{Manolis Vasilakis}
{Universit\'{e} Paris-Dauphine, PSL University, CNRS UMR7243, LAMSADE, Paris, France}
{emmanouil.vasilakis@dauphine.eu}
{https://orcid.org/0000-0001-6505-2977}
{Partially supported by ANR project ANR-21-CE48-0022 (S-EX-AP-PE-AL).}

\authorrunning{T. Gima,  E. J. Kim, N. K\"{o}hler, N. Melissinos, and M. Vasilakis}

\Copyright{Tatsuya Gima, Eun Jung Kim, Noleen K\"{o}hler, Nikolaos Melissinos, and Manolis Vasilakis}



\category{} 

\funding{Our research visit to Nagoya University, Japan was funded by the PRC CNRS JSPS2019-2020 program, project PARAGA (Parameterized Approximation Graph Algorithms).}

\acknowledgements{We would like to thank  Virginia Ard\'{e}vol Mart\'{i}nez and Yota Otachi for interesting discussions at the preliminary stages of this work.}

\EventEditors{Neeldhara Misra and Magnus Wahlstr\"{o}m}
\EventNoEds{2}
\EventLongTitle{18th International Symposium on Parameterized and Exact Computation (IPEC 2023)}
\EventShortTitle{IPEC 2023}
\EventAcronym{IPEC}
\EventYear{2023}
\EventDate{September 6--8, 2023}
\EventLocation{Amsterdam, The Netherlands}
\EventLogo{}
\SeriesVolume{285}
\ArticleNo{22}

\begin{document}

\maketitle

\begin{abstract}
Given a graph $G$ and an integer $b$,
\textsc{Bandwidth} asks whether there exists a bijection $\pi$ from $V(G)$ to $\{1, \ldots, |V(G)|\}$
such that $\max_{\{u, v \} \in E(G)} | \pi(u) - \pi(v) | \leq b$.
This is a classical NP-complete problem, known to remain NP-complete even on very restricted classes of graphs,
such as trees of maximum degree 3 and caterpillars of hair length 3.
In the realm of parameterized complexity, these results imply that the problem remains NP-hard on graphs of bounded pathwidth,
while it is additionally known to be W[1]-hard when parameterized by the tree-depth of the input graph.
In contrast, the problem does become FPT when parameterized by the vertex cover number.
In this paper we make progress in understanding the parameterized (in)tractability of \textsc{Bandwidth}.
We first show that it is FPT when parameterized by the cluster vertex deletion number cvd plus the clique number $\omega$,
thus significantly strengthening the previously mentioned result for vertex cover number.
On the other hand, we show that \textsc{Bandwidth} is W[1]-hard when parameterized only by cvd.
Our results develop and generalize some of the methods of argumentation of the previous results and narrow some of the complexity gaps.
\end{abstract}

\section{Introduction}\label{sec:introduction}

Given an undirected graph $G$ and an integer $b$,
{\BW} asks whether there exists a bijection $\pi \colon V(G) \to \{ 1, \ldots, |V(G)| \}$ of the vertices of $G$ (called an \emph{ordering})
such that $\max_{\{u,v\} \in E(G)} \abs{\pi(u) - \pi(v)} \leq b$.
The main motivation behind its study dates back to over half a century;
a closely related problem in the field of matrix theory was first studied in the 1950's,
while in the 1960's it was formulated as a graph problem, finding applications in minimizing (average) absolute error in codes,
and has been extensively studied ever since~\cite{jgt/ChinnCDG82,tcs/CyganP10,dam/CyganP12,talg/CyganP12,jcss/DubeyFU11,algorithmica/FeigeT09,tcs/FurerGK13,jal/Gupta01,Harper64}.

{\BW} is long known to be NP-complete~\cite{jgt/ChinnCDG82,Papadimitriou76};
as a matter of fact, that is the case even on very restricted classes of graphs,
such as trees of maximum degree 3~\cite{GareyGJK78},
caterpillars of hair length 3~\cite{Burhardk86},
and cyclic caterpillars of hair length 1~\cite{tcs/Muradian03}.
Considering these NP-hardness results, in this paper we focus on the problem's parameterized complexity.
When parameterized by the natural parameter $b$, {\BW} is known to be in XP~\cite{jal/GurariS84,siammax/Saxe80},
whilst it is XNLP-complete even when the input graph is a tree~\cite{wg/Bodlaender21,iandc/BodlaenderGNS24},
which implies W[$t$]-hardness for all positive integers $t$.
In fact, {\BW} cannot be solved in time $f(b)n^{o(b)}$ even on trees of pathwidth at most two
unless the Exponential Time Hypothesis fails~\cite{icalp/DregiL14},
rendering the known $2^{\bO(b)} n^{b+1}$ algorithm essentially optimal~\cite{siammax/Saxe80}.
Regarding structural parameterizations,
the previously mentioned results imply that {\BW} is para-NP-complete when parameterized by the pathwidth or the treewidth plus the maximum degree of the input graph;
the latter implies NP-completeness also on graphs of constant tree-cut width~\cite{siamdm/GanianKS22}.
Moreover, it is known to be W[1]-hard parameterized by the tree-depth of the input graph~\cite{tcs/GimaHKKO22}.
In contrast, the problem does become fixed-parameter tractable (FPT) when parameterized by the vertex cover number~\cite{isaac/FellowsLMRS08},
the neighborhood diversity~\cite{mthesis/Bakken18}, or the max-leaf number~\cite{mst/FellowsLMMRS09}.

In the last few years, a plethora of structural parameters have been introduced
in an attempt to precisely determine the limits of tractability of algorithmic problems that are FPT parameterized by vertex cover number,
yet become W[1]-hard when parameterized by more general parameters, such as treewidth or clique-width.
Some of the most well-studied such parameters are
vertex integrity~\cite{tcs/GimaHKKO22},
tree-depth~\cite{ejc/NesetrilM06},
twin cover number~\cite{dmtcs/Ganian15},
cluster vertex deletion number~\cite{mfcs/DouchaK12},
shrub-depth~\cite{lmcs/GanianHNOM19},
neighborhood diversity~\cite{algorithmica/Lampis12},
and modular-width~\cite{iwpec/GajarskyLO13}.
The tractability of {\BW} with respect to those parameters has remained largely unexplored,
with the exception of neighborhood diversity~\cite{mthesis/Bakken18} and tree-depth~\cite{tcs/GimaHKKO22}.

Cluster vertex deletion number lies between clique-width and vertex cover number (more precisely twin cover number)
and is defined as the minimum size of a set of vertices whose removal induces a cluster graph,
i.e., a disjoint union of cliques.
Computing the cluster vertex deletion number of a graph is known to lie in FPT~\cite{mst/BoralCKP16,mst/HuffnerKMN10,mst/Tsur21} and
it was first considered as a structural parameter in~\cite{mfcs/DouchaK12},
while it has been used to show parameterized (in)tractability results in multiple occasions
ever since~\cite{jcss/BanikKR25,algorithmica/BruhnCJS16,tcs/ChlebikovaC17,mst/ChopinNNW14,iwoca/KareR19,algorithmica/KuceraS23,algorithmica/MajumdarR18}.
Notice that {\BW} is trivial on cluster graphs;
it suffices to check whether the clique number is at most $b+1$,
as there exists an optimal ordering that places the vertices of every clique consecutively,
for some ordering of the cliques.
Therefore, its tractability when parameterized by the cluster vertex deletion number of the input graph
poses a very natural question.

\begin{figure}[tb]
    \centering
    \tikzset{
    fpt/.style={
        shape=rectangle,
        fill=green!40!white,
        draw,
        line width= .2,
        rounded corners = 3pt,
        inner sep = 2,
        outer sep = 2,
        minimum height = 12,
    },
    w1/.style={
        shape=rectangle,
        fill=orange!40!white,
        rounded corners = 3pt,
        inner sep = 2,
        outer sep = 2,
        minimum height = 12,
    },
    nph/.style={
        shape=rectangle,
        fill=red!40!white,
        rounded corners = 1pt,
        inner sep = 3,
        outer sep = 3,
        draw,
        double,
        double distance=1,
        minimum height = 12,
    }
}

\begin{tikzpicture}[every node/.style={thick, align=center}, scale=0.9]
  \small

  \node[fpt] (vc) at (0,   0) {vertex cover number~\cite{isaac/FellowsLMRS08}};
  \node[nph,above = 6.5 of vc] (cw) {clique-width};

  \node[fpt,above = .3 of vc] (tco)  {twin cover number $\operatorname{+} \omega$};
  \node[above = .8 of tco] (anccd) {};

  \node[above left = .05 of anccd] (tc)  {twin cover number};
  \node[fpt,right = 1.2  of anccd.north] (cdo) {cluster vertex\\deletion number $\operatorname{+} \omega$ \\\relax[\cref{thm:FPTAlgo}]};

  \node[above = 0.3 of cdo] (vi)  {vertex integrity};
  \node[w1, above = 0.4 of vi] (td)  {tree-depth~\cite{tcs/GimaHKKO22}};
  \node[w1, right = 2.5 of vi.south] (bw) {bandwidth~\cite{iandc/BodlaenderGNS24}};
  \node[fpt,below = 1 of bw] (ml) {max-leaf number~\cite{mst/FellowsLMMRS09}};

  \node[nph,above = .4 of td] (pw) {pathwidth~\cite{Burhardk86}};

  \node[nph,above = .4 of pw] (tw) {treewidth};
  \node[nph,above =  .7of bw] (twD) {treewidth + $\Delta$~\cite{GareyGJK78}};
  \node[nph,above  =  .5of twD] (tcw) {tree-cut width};

  \node[w1,above = of anccd] (cd) {cluster vertex\\ deletion number\\\relax[\cref{thm:hardness}]};
  \node[w1,above = of cd] (sd) {shrub-depth};

  \node[fpt,above left = 1.3 of vc] (nd) {neighborhood\\diversity~\cite{mthesis/Bakken18}};
  \node[above = 2.3 of nd] (mw) {modular-width};

  \draw[thick,<-] (cw.east) -- (tw.north west);
  \draw[thick,<-] (tw) -- (pw);
  \draw[thick,<-] (tw.east) -- (tcw.north west);
  \draw[thick,<-] (pw) -- (td);
  \draw[thick,<-] (twD) -- (bw);
  \draw[thick,<-] (tcw) -- (twD);
  
  \draw[thick,<-] (pw.south east) -- (bw.north west);
  \draw[thick,<-] (bw) -- (ml);

  \draw[thick,<-] (sd) -- (cd);
  \draw[thick,<-] (sd) -- (td.north west);
  \draw[thick,<-] (sd) -- (nd);
  \draw[w1,thick,<-] (cw) -- (sd);

  \draw[thick,<-] (cd) -- (cdo.north west);
  \draw[thick,<-] (cdo.south west) -- (tco);
  \draw[thick,<-] (tco) -- (vc);
  \draw[thick,<-] (tc) -- (tco);
  \draw[thick,<-] (td) -- (vi);
  \draw[thick,<-] (vi) -- (cdo);
  \draw[thick,<-] (cd) -- (tc);
  
  \draw[thick,<-] (cw.south west) -- (mw);
  \draw[thick,<-] (mw) -- (tc);
  \draw[thick,<-] (mw) -- (nd);
  \draw[thick,<-] (nd) -- (vc.north west);

\end{tikzpicture}
    \caption{Our results and hierarchy of some related structural graph parameters,
    where $\omega$ and $\Delta$ denote the clique number and the maximum degree of the input graph, respectively.
    Arrows between parameters indicate generalization relations, that is, for any graph,
    if the parameter at the tail of an arrow is a constant then the parameter at the head of the arrow is also a constant.
    The reverse does not hold in this figure.
    The framed green, frameless orange, and double framed red rectangles indicate
    fixed-parameter tractable, W[$\ast$]-hard, and NP-complete cases, respectively.
    }
    \label{fig:intro:results}
\end{figure}
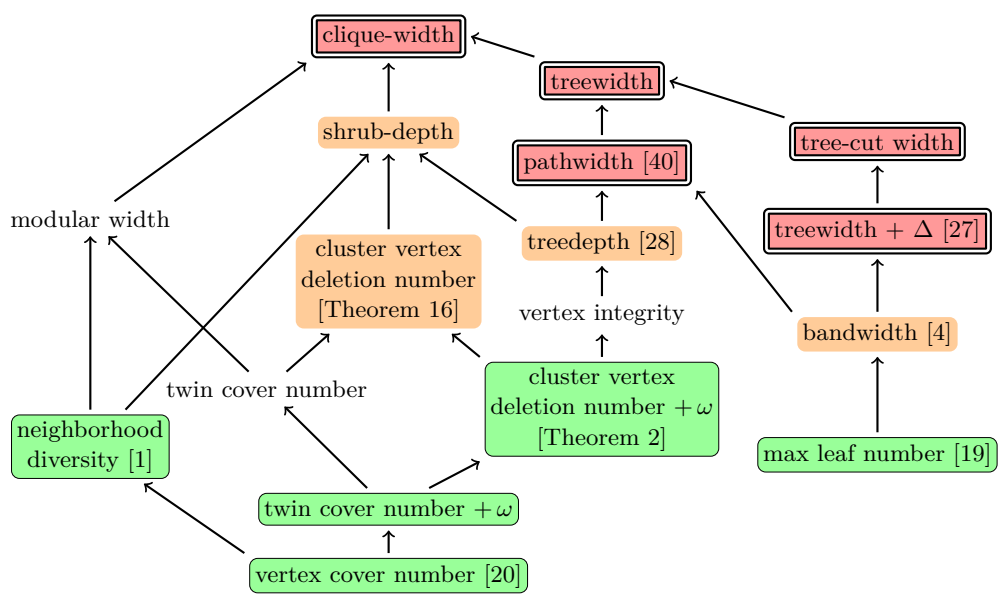

\subparagraph*{Our contribution.}
In the current work, we present both tractability and intractability results for {\BW} when cluster vertex deletion number is a parameter of the problem
(see \cref{fig:intro:results} for an overview of our results and the relationships between the structural parameters mentioned).
We first prove that {\BW} is FPT when parameterized by $\clusterDel + \omega$,
where $\clusterDel$ and $\omega$ denote the cluster vertex deletion number and clique number of the input graph respectively.
This significantly strengthens the tractability result for vertex cover number of~\cite{isaac/FellowsLMRS08},
and follows the same idea of encoding the problem as an \emph{integer linear program} (ILP) of a small number of variables.
Solving said ILP, one can verify whether there exists any ordering $\pi$ of the vertices of $G$ such that
(a) $|\pi(v) - \pi(u)| \leq b$ for all $\braces{u,v} \in E(G)$, and
(b) $\pi$ is ``nice'', where an ordering is nice if it has some specific properties.
Proving that there exists a nice ordering $\pi$ that minimizes $\max_{\braces{u,v} \in E(G)} |\pi(v) - \pi(u)|$ then yields the stated result.

A natural question that arises from the previous result is whether it is necessary for both $\clusterDel$ and $\omega$ to be parameters of the problem in order to assure fixed-parameter tractability.
Notice that {\BW} is NP-complete even when $\omega \leq 2$, since that is the case for trees.
Therefore, we proceed by studying the problem's tractability when parameterized only by $\clusterDel$.
In this setting, we show that {\BW} is W[1]-hard via a reduction from \UBP,
thus positively answering the previous question.

\subparagraph*{Related work.}
{\BW} is one of the so-called \emph{graph layout} problems (see the survey of~\cite{csur/DiazPS02}).
As far as the structural parameterized complexity of such problems is concerned,
Fellows, Lokshtanov, Misra, Rosamond, and Saurabh~\cite{isaac/FellowsLMRS08} were the first to prove FPT results for a multitude of them
when parameterized by the vertex cover number of the input graph, making use of ILP formulations.
Since then, not much progress has been made on that front, with a notable exception being \textsc{Imbalance},
which was shown to be FPT when parameterized by twin cover number plus $\omega$~\cite{tcs/MisraM21},
vertex integrity~\cite{tcs/GimaHKKO22}, or tree-cut width~\cite{siamdm/GanianKS22},
while it belongs to XP when parameterized by twin cover~\cite{tcs/MisraM21}.
\textsc{Minimum Linear Arrangement} is known to be FPT parameterized by max-leaf number, or edge clique number of the input graph~\cite{toct/FellowsHRS16},
as well as by the vertex cover number~\cite{arxiv/Lokshtanov17}.
Lastly, as far as \textsc{Cutwidth} is concerned,
a $2^{\bO(\vc)} n^{\bO(1)}$ time algorithm was presented in~\cite{algorithmica/CyganLPPS14},
improving over the ILP formulation of~\cite{isaac/FellowsLMRS08},
where $\vc$ denotes the vertex cover number.

\subparagraph*{Organization.}
In \cref{sec:preliminaries} we discuss the general preliminaries,
followed by the FPT-algorithm in \cref{sec:parameterizedTractability}
and the hardness result in \cref{sec:parametrizedHardness}.
Lastly, in \cref{sec:conclusion} we present the conclusion as well as some directions for future research.


\section{Preliminaries}\label{sec:preliminaries}

Throughout the paper we use standard graph notation~\cite{Diestel17},
and we assume familiarity with the basic notions of parameterized complexity~\cite{books/CyganFKLMPPS15}.
We assume that $\mathbb{N}$ is the set of all non-negative integers.
All graphs considered are undirected without loops.
The \emph{clique number} of a graph $G$, denoted by $\omega(G)$, is the size of its largest induced clique.
For $x, y \in \Z$, let $[x, y] = \setdef{z \in \Z}{x \leq z \leq y}$,
while $[x] = [1,x]$.
For $\mathcal{I}_i = [a_i, b_i]$, we say that intervals $\mathcal{I}_1, \ldots, \mathcal{I}_k$ \emph{partition} interval $\mathcal{I} = [a, b]$
if $\mathcal{I} = \bigcup_{i \in [k]} \mathcal{I}_i$ and $\mathcal{I}_i \cap \mathcal{I}_j = \varnothing$, for any $1 \leq i < j \leq k$.
Additionally, let $\mathcal{I}_i < \mathcal{I}_j$ if $b_i < a_j$.%
For a function $f \colon A \to B$ and $A' \subseteq A$,
let $f(A') = \setdef{f(a) \in B}{a \in A'}$.
Moreover, let $\maxofset{f}{A'} = \max \setdef{f(a)}{a \in A'}$ and $\minofset{f}{A'}$ defined analogously.

Let $G$ be a graph and $\pi \colon V(G) \to [n]$ an ordering of its vertices.
We define the \emph{stretch} of an edge $e = \braces{u, v} \in E(G)$ with regard to $\pi$ as $\stretchEdge{\pi}(e) = |\pi(u) - \pi(v)|$.
We define the \emph{stretch} of $\pi$ to be the maximum stretch of the edges of $G$ with regard to $\pi$,
i.e., $\str(\pi) = \max_{e \in E(G)} \stretchEdge{\pi}(e)$.
The \emph{bandwidth} of $G$, denoted $\bw(G)$,
is the minimum stretch over all its possible vertex orderings,
while {\BW} asks, given a graph $G$ and an integer $b$, whether $\bw(G) \le b$.

\begin{remark}\label{rem:automorphism}
    Note that the stretch of a vertex ordering is invariant under isomorphism,
    which means in particular that $\str(\pi) = \str(\pi \circ f)$ for any vertex ordering $\pi \colon V(G) \to [n]$
    and any automorphism $f \colon V(G) \to V(G)$ of $G$.
\end{remark}

A \emph{cluster deletion set} of a graph $G$ is a set $S \subseteq V(G)$ such that every component of $G-S$ is a clique.
If $S$ is a cluster deletion set, we call the components of $G-S$ \emph{clusters}.
The \emph{cluster vertex deletion number} of $G$, denoted $\clusterDel(G)$, is the size of its minimum cluster deletion set.

The feasibility variant of \emph{integer linear programming} (ILP) is to decide,
given a set $X$ of variables and a set $C$ of linear constraints (i.e., inequalities)
over the variables in $X$ with integer coefficients,  whether there is an assignment $\alpha \colon X \to \Z$
of the variables satisfying all constraints in $C$.
It is known that the feasibility of an instance of (ILP) can be tested in $\bO\parens{p^{2.5p+o(p)} \cdot L}$ time,
where $p$ is the number of variables and $L$ is the size of the input~\cite{FrankT87,Kannan87,Lenstra83}. 
In other words, computing the feasibility of an ILP formula is FPT parameterized by the number of variables.
Moreover, a solution can be computed in the same time if it exists.




\section{An FPT-algorithm parameterized by cluster vertex deletion number plus clique number}\label{sec:parameterizedTractability}

In this section we prove that {\BW} is FPT parameterized by the cluster vertex deletion number plus the clique number of the input graph.
We first state the main theorem of this section, followed by a short discussion regarding the considered parameterization
and a high-level overview of our proof.

\begin{theorem}\label{thm:FPTAlgo}
    {\BW} is fixed-parameter tractable when parameterized by $\clusterDel + \omega$,
    where $\clusterDel$ and $\omega$ denote the cluster vertex deletion number and clique number of the input graph respectively.
\end{theorem}

\subparagraph*{Parameterization by $\clusterDel + \omega$.}
Given a graph $G$ and an integer $k$, it is known that one can compute a cluster deletion set of $G$ of size at most $k$ in time $2^{\bO(k)} n^{\bO(1)}$,
or correctly determine that no such set exists~\cite{mst/BoralCKP16,mst/HuffnerKMN10,mst/Tsur21}.
Notice that one can use this algorithm in order to determine, for a graph $G$ and an integer $k$,
whether $\clusterDel(G) + \omega(G) \le k$ in time $2^{\bO(k)} n^{\bO(1)}$:
(i) run the previously mentioned algorithm for $i = 1, \ldots, \min\{ \clusterDel(G), k \}$ to compute a minimum cluster deletion set $S$ of size $|S| = \clusterDel(G)$,
provided such a set exists,
(ii) compute $\omega(G)$ by guessing the subset $S' \subseteq S$ which is the intersection of a maximum clique with $S$,
and then consider for each cluster $C$ of $G-S$ the graph induced by the vertices of $C$ and $S'$;
the considered maximum clique is then composed of $S'$ as well as the vertices of a cluster $C$
that are adjacent to all the vertices of $S'$.

\subparagraph*{Proof overview.}
Our proof is a generalization of the FPT result for vertex cover number from~\cite{isaac/FellowsLMRS08}.
The general idea for obtaining an ILP encoding of {\BW} given a vertex cover $S$ is to augment $S$ by a small number
(dependent only on the vertex cover number) of representative vertices of every neighborhood-type.
It can be easily seen that we can modify any ordering $\pi$ in such a way that
the leftmost and rightmost neighbor of any vertex in $S$ is contained in this augmented set $S'$ without increasing the stretch.
For any ordering $\sigma$ of $S'$ we can decide whether we can extend $\sigma$ to an ordering of $V(G)$ of stretch at most $b$ by
encoding how vertices of certain neighborhood-types are distributed into the gaps between the vertices of $S'$ into an ILP.
By ensuring that we distribute the vertices in such a way that the leftmost and rightmost neighbor of any vertex in $S$ is contained
in $S'$ we can bound the stretch of every edge by using one linear constraint for every edge in $G[S']$. 

In our setting we can use the vertex cover approach to bound the stretch of all edges incident to the deletion set.
To gain control over the stretch of edges within clusters, we show that we can convert any ordering into a \emph{nice} one without increasing the stretch.
Here niceness intuitively means, that we can order the vertices in between any two vertices of $S'$ in such a way that vertices of the same type appear consecutively,
where the type now depends on the isomorphism-type of the cluster union the deletion set. 
This will allow us to bound the stretch of such edges by a linear constraint as well.

\subsection{Types and buckets}

Throughout this subsection we let $G$ be a graph and $S$ a cluster deletion set of $G$ of size $k$.
For any vertex $v \in V(G-S)$, let $N_S(v) = N(v) \cap S$ be its \emph{$S$-neighborhood}. 
In the following we define types of clusters.
Cluster-types capture the isomophism-type of the union of a cluster with the deletion set $S$.
As clusters of the same size are isomorphic, our notion of type only needs to capture how the vertices of the cluster interact with the deletion set.
In particular, we define cluster-types to be vectors such that each entry specifies the number of vertices in the cluster which have a particular $S$-neighborhood;
in that case, two clusters are of the same cluster-type if and only if they agree on the number of their vertices with $S$-neighborhood $N$, for all $N \subseteq S$.
We define this formally as follows.
Let $\clusterTypes \subseteq \mathbb{N}^{2^k}$ be the set of non-negative integer vectors $\kappa$ with $2^k$
entries for which $\lVert \kappa \rVert_1 \leq \omega(G)$.%
\footnote{Here $\lVert \cdot \rVert_1$ denotes the $1$-norm, i.e., the sum of the absolute values of the entries.}
We assume that the entries of the vectors in $\clusterTypes$ are indexed by the subsets of $S$.
We say that a cluster $C$ has \emph{cluster-type} $\kappa \in \clusterTypes$ if $|\setdef{v \in V(C)}{N_S(v)=N}|=(\kappa)_N$ for every $N\subseteq S$ where $(\kappa)_N$ denotes the entry of $\kappa$ corresponding to $N$.
We further let $\#\kappa$ denote the number of clusters of cluster-type $\kappa$ in $G-S$.
We say that a set $\mathcal{C}$ of clusters is \emph{representative} if it consists of
$\min \braces{2 |S|, \#\kappa}$ distinct clusters of type $\kappa$ for every cluster-type $\kappa$.
We further say that a set $S'$ is an \emph{extended deletion set} if $S' = S \cup \bigcup_{C \in \mathcal{C}} V(C)$ for a representative set $\mathcal{C}$ of clusters.

\begin{lemmarep}\label{lem:stretchSClusters}
    For every extended deletion set $S'$ there is an ordering $\pi \colon V(G) \to [n]$ such that 
    $\str(\pi) = \bw(G)$ and for every $s \in S$, the set $S'$ contains vertices $v_{\min}^s$ and $v_{\max}^s$,
    where $\pi(v_{\min}^s) = \minofset{\pi}{N(s)}$ and
    $\pi(v_{\max}^s) = \maxofset{\pi}{N(s)}$,
    i.e., $S'$ contains the leftmost and rightmost neighbor of $s$.    
\end{lemmarep}

\begin{proof}
    Let $S'$ be an extended deletion set, where $S' = S \cup \bigcup_{C \in \mathcal{C}} V(C)$ for some representative cluster set $\mathcal{C}$,
    and $\pi$ an ordering of $V(G)$ such that $\str(\pi) = \bw(G)$.
    Notice that since $|\setdef{v^s_{\min}, v^s_{\max}}{s \in S}| \leq 2 |S|$,
    it follows that there are at most $2 |S|$ vertices among them not belonging to $S$ and belonging to clusters of the same cluster-type.

    Let $s \in S$ and suppose that $v^s \in \braces{v^s_{\min}, v^s_{\max}}$ does not belong to $S'$.
    Let $C \notin \mathcal{C}$ denote the cluster of $G-S'$ that contains $v^s$, and assume that it is of cluster-type $\kappa$.
    In that case, it follows that $\#\kappa > 2 |S|$ and that there exists a cluster $C' \in \mathcal{C}$ of cluster-type $\kappa$,
    such that $V(C') \cap \setdef{v^s_{\min}, v^s_{\max}}{s \in S} = \varnothing$.
    Exchanging the vertices of $C$ with vertices of $C'$ respecting their $S$-neighborhood yields an automorphism $f \colon V(G) \to V(G)$.    
    For $\pi' = \pi \circ f$, it holds that $\str(\pi') = \bw(G)$ due to \cref{rem:automorphism},
    while $v^s \in S'$.
    By repeatedly applying the previous argument, the statement follows.
\end{proof}

We say that an ordering $\pi \colon V(G) \to [n]$ is \emph{$S'$-extremal} if the second property in \cref{lem:stretchSClusters} is satisfied for $\pi$.

\begin{observation}
    Since there are at most $\omega(G) 2^{|S| \cdot \omega(G)}$ different cluster-types,
    it holds that the size of any extended deletion set $S'$ is at most $|S| + 2 |S| \omega(G) \cdot \omega(G) 2^{|S| \cdot \omega(G)}$.
\end{observation}

Let $\mathcal{C}$ be a representative set of clusters,
$S' = S \cup \bigcup_{C \in \mathcal{C}} V(C)$ the extended deletion set containing vertices from $\mathcal{C}$ and $S$, and $k' = |S'|$.
A \emph{bucket distribution} of $S'$ is a partition $\mathcal{B} = (B_0, \ldots, B_{k'})$ of the vertices of $G-S'$.
Fix a bucket distribution $\mathcal{B} = (B_0, \ldots, B_{k'})$ of $S'$.
We call the subsets $B_i$ \emph{buckets} of $\mathcal{B}$.

Having fixed a bucket distribution, we define a more refined notion of types of clusters.
In this refined notion, which we call distribution-type,
we specify how many vertices of a cluster with a particular $S$-neighborhood are contained in a particular bucket.
We define this formally in the following.
Let $\distributionTypes \subseteq \mathbb{N}^{2^k \times (k'+1)}$ be the set of matrices $\tau$ with $\lVert \tau \rVert_1\leq \omega(G)$.
We assume that the rows of matrices are indexed with subsets of $S$ and the columns with $[0,k']$.
We say that a cluster $C \notin \mathcal{C}$ has \emph{distribution-type} $\tau \in \distributionTypes$ in  $\mathcal{B}$
if $|\setdef{v \in V(C) \cap B_i}{N_S(v) = N}| = (\tau)_{N,i}$ for every $N \subseteq S$ and every $i \in [0,k']$.
For every $\kappa \in \clusterTypes$, let $\distributionTypes_\kappa \subseteq \distributionTypes$ denote the set of distribution-types $\tau$ such that
$\sum_{i\in[0,k']} (\tau)_{N,i} = (\kappa)_N$ for every $N \subseteq S$,
i.e., the set of $\tau \in \distributionTypes$ such that any cluster of distribution-type $\tau$ has cluster-type $\kappa$.

\begin{observation}\label{obs:sizeOfExtended}
      The number of distribution-types is at most ${\omega(G)}^{2^{|S|} \cdot (|S'| + 1)}$ for any extended deletion set $S'$.
\end{observation}

Let $\sigma \colon S' \to [k']$ be an ordering of the vertices of $S'$.
We say that a vertex ordering $\pi \colon V(G) \to [n]$ is \emph{compatible} with $\sigma$ if
for any $s_1, s_2 \in S'$ it holds that $\pi(s_1) < \pi(s_2)$ if and only if $\sigma(s_1) < \sigma(s_2)$.
We say that a vertex ordering $\pi \colon V(G) \to [n]$ is \emph{compatible} with $\sigma$
and $\mathcal{B}$ if $\pi$ is compatible with $\sigma$ and 
$B_0 = \setdef{v \in V(G)}{\pi(v) < \pi(\sigma^{-1}(1))}$, $B_{k'} = \setdef{v \in V(G)}{\pi(v) > \pi(\sigma^{-1}(k'))}$ and
$B_i = \setdef{v \in V(G)}{\pi(\sigma^{-1}(i)) < \pi(v) < \pi(\sigma^{-1}(i+1))}$ for $i \in [k'-1]$.

\subsection{Nice orderings}\label{subsec:niceOrderings}
Throughout this subsection we let $G$ be a graph and $S$ a cluster deletion set of $G$.
Furthermore, for this subsection we let  $\mathcal{C}$ be a representative set  of clusters,
$S' = S \cup \bigcup_{C \in \mathcal{C}} V(C)$ the extended deletion set containing vertices from $\mathcal{C}$ and $S$, and $k' = |S'|$.
Additionally, we fix a bucket distribution $\mathcal{B}=(B_0, \ldots, B_{k'})$ of $S'$ and an ordering $\sigma \colon S' \to [k']$.

To obtain our nice ordering we use a series of exchange arguments that will not increase the stretch.
We call an ordering \emph{nice} if it has Properties~\hyperref[property:1]{$(\Pi_1)$},~\hyperref[property:2]{$(\Pi_2)$}, and~\hyperref[property:3]{$(\Pi_3)$}.
We will first give some intuition regarding the properties, before defining them formally.

Assume that $\pi \colon V(G) \to [n]$ is an optimal ordering minimizing the number of edges of maximum stretch and
assume that $\pi$ is compatible with $\sigma$ and $\mathcal{B}$.
Furthermore, let  $v\in V(C)$ be a vertex which is contained in an edge of maximum stretch with regard to $\pi$ and the cluster $C$ containing $v$ is distributed over more than one bucket.
In this case, $v$ must be either the leftmost or the rightmost vertex of $C$.
Assuming $v \in B_i$ is the leftmost vertex of $C$ (the other case is analogous),
we can observe that every vertex $v' \in B_i$ appearing further to the right than $v$ must have a neighbor contained in a bucket to the right of $B_i$ and no neighbor to the left of $v$.
Otherwise, we can reduce the stretch of the edge containing $v$ without increasing the stretch of any edge incident to $v'$
(and hence reducing the number of edges of maximum stretch without increasing the maximum stretch) by exchanging $v$ and $v'$. 
Therefore, we can assume that each bucket is partitioned into a left, a middle, and a right part
and every vertex having neighbors only to the buckets to the left (respectively, right) of $B_i$ appears in the left (respectively, right) part.
Additionally, this allows us to assume that within each bucket the vertices of one cluster appear consecutively (Property~\hyperref[property:1]{$(\Pi_1)$}).

Now assume that $\{v,w\}$ is an edge of maximum stretch as before ($v$ appears left of $w$ in $\pi$) and
$\{v',w'\}$ is another edge such that $v'$ appears in the same bucket as $v$ and $w'$ in the same bucket as $w$.
If $v'$ appears before $v$ then $w'$ has to appear before $w$ as $\{v,w\}$ is of maximum stretch.
On the other hand, if $v'$ appears after $v$ then $w'$ must appear after $w$ as otherwise exchanging $w$ and $w'$ either reduces
the number of edges of maximum stretch or reduces the maximum stretch itself.
Hence, we can assume that the relative order of the leftmost vertices of a set of clusters is the same as the relative order of the rightmost vertices of the same clusters (Property~\hyperref[property:2]{$(\Pi_2)$}).

Lastly, assume that $C$ and $C'$ are clusters of type $\tau \in \distributionTypes$ which are not contained in just one bucket and appear next to each other (in their leftmost bucket).
Assume $B_\ell$ is the bucket containing the leftmost vertex of $C$ and $C'$ and $B_r$ the bucket containing the rightmost vertex of $C$ and $C'$.
We can essentially exchange $V(C) \cap B_\ell$ with $V(C') \cap B_\ell$ and at the same time $V(C) \cap B_r$ with $V(C') \cap B_r$ if certain properties about the size of these sets hold.
This allows us to order the buckets in such a way, that clusters whose intersection with the leftmost (rightmost, respectively) bucket they intersect is of the same size, appear consecutively (Property~\hyperref[property:3]{$(\Pi_3)$}).

To state Properties~\hyperref[property:1]{$(\Pi_1)$},~\hyperref[property:2]{$(\Pi_2)$}, and~\hyperref[property:3]{$(\Pi_3)$}
formally we use the following notation.
For a distribution-type $\tau \in \distributionTypes$ and $i \in [0,k']$,
we write $\tau_i$ to denote the column of $\tau$ which is indexed by $i$.
We define $\leftBucket(\tau)$ to be the largest index $i \in [0,k']$ such that $\lVert \tau_j \rVert_1 = 0$ for any $j \in [0, i-1]$,
i.e., $B_i$ is the leftmost bucket containing vertices from clusters of type $\tau$.
We define $\rightBucket(\tau)$ analogously to be the minimum index $i \in [0,k']$ such that $\lVert \tau_j \rVert_1 = 0$ for any $j \in [i+1, k']$.
Additionally, we let $\leftSize(\tau)$ be $\lVert \tau_{\leftBucket(\tau)} \rVert_1$ and $\rightSize(\tau)$ be $\lVert \tau_{\rightBucket(\tau)} \rVert_1$.
For every $\ell \leq r \in [0,k']$  and every $n_L,  n_R \in [0,\omega(G)]$,
we define
\[
    \distributionTypes^{(\ell,r,n_L,n_R)} =
    \setdef{\tau\in \distributionTypes}{\leftBucket(\tau)=\ell, \, \rightBucket(\tau)=r, \, \leftSize(\tau)=n_L, \, \rightSize(\tau)=n_R}.
\]

\begin{definition}[Property $(\Pi_1)$]\label{property:1}
    We say that an $S'$-extremal ordering $\pi \colon V(G) \to [n]$ which is compatible
    with $\sigma$ and $\mathcal{B}$ has \emph{Property~\hyperref[property:1]{$(\Pi_1)$}} if for every $i \in [0,k']$
    \begin{enumerate}
        \item the vertices of $V(C) \cap B_i$ appear consecutively in $\pi$ for every cluster $C \notin \mathcal{C}$,

        \item we can partition the interval $\pi(B_i)$ into three (possibly empty) intervals $I_R^{i} < I_M^{i} < I_L^{i}$
        such that for every $\tau \in \distributionTypes$ and every cluster $C$ of distribution-type $\tau$
        \begin{itemize}
            \item $\pi(V(C) \cap B_i) \subseteq I_R^{i}$ if $\leftBucket(\tau) \neq i$ and $\rightBucket(\tau) = i$, 
            
            \item $\pi(V(C) \cap B_i) \subseteq I_L^{i}$ if $\leftBucket(\tau) = i$ and $\rightBucket(\tau) \neq i$, 
            
            \item $\pi(V(C) \cap B_i) \subseteq I_M^{i}$ if either $\leftBucket(\tau) \neq i$ and $\rightBucket(\tau) \neq i$,
                or $\leftBucket(\tau) = \rightBucket(\tau) = i$.
        \end{itemize}
    \end{enumerate}
\end{definition}
Notice that while $I_R^{i}$ contains the leftmost ordered vertices of $B_i$,
we use the index $R$ since those vertices are the rightmost vertices of their corresponding cliques.
Analogously, we use $I_L^{i}$ for the rightmost ordered vertices of $B_i$.

\begin{toappendix}
\begin{lemma}\label{lem:leftMiddleRightPartition_and_consecutiveCliques}
    Given an $S'$-extremal ordering $\pi \colon V(G) \to [n]$ which is compatible with $\sigma$ and $\mathcal{B}$,
    there exists an $S'$-extremal ordering $\pi' \colon V(G) \to [n]$ with $\str(\pi') \leq \str(\pi)$
    which is compatible with $\sigma$ and $\mathcal{B}$ and has Property~\hyperref[property:1]{$(\Pi_1)$}.
\end{lemma}

\begin{proof}
    Let $\pi \colon V(G) \to [n]$ be an $S'$-extremal ordering of $V(G)$ which is compatible with $\sigma$ and $\mathcal{B}$.
    For every $i \in [0,k']$, let $R_i \subseteq B_i$ contain the vertices that belong to the clusters of distribution-type $\tau$,
    for which it holds that $\leftBucket(\tau) \neq i$ and $\rightBucket(\tau) = i$.
    Additionally, let $L_i \subseteq B_i$ contain the vertices that belong to the clusters of distribution-type $\tau$,
    for which it holds that $\leftBucket(\tau) = i$ and $\rightBucket(\tau) \neq i$.
    Lastly, let $M_i = B_i \setminus (L_i \cup R_i)$,
    i.e., $M_i$ contains the vertices of $B_i$ that belong to the clusters of distribution-type $\tau$,
    for which it holds that either $\leftBucket(\tau) \neq i$ and $\rightBucket(\tau) \neq i$, or $\leftBucket(\tau) = \rightBucket(\tau) = i$.

    Fix an $i \in [0,k']$, and for every cluster $C$ with $V(C) \cap B_i \neq \varnothing$,
    let $v^{C,i}_{\min}$ and $v^{C,i}_{\max}$ denote its leftmost and rightmost vertex in $B_i$ on ordering $\pi$,
    i.e., $\pi(v^{C,i}_{\min}) = \minofset{\pi}{V(C) \cap B_i}$ and $\pi(v^{C,i}_{\max}) = \maxofset{\pi}{V(C) \cap B_i}$.

    We first define an ordering of the vertices of $R_i$.
    Let $C_1, \ldots, C_\ell$ denote the cliques of $G-S'$ that have vertices in $R_i$,
    where $\maxofset{\pi}{V(C_j) \cap R_i} < \maxofset{\pi}{V(C_{j+1}) \cap R_i}$,
    for all $j \in [\ell - 1]$.
    Then, define $\pi_{R_i} \colon R_i \to [|R_i|]$ as follows:
    \begin{itemize}
        \item $\pi_{R_i} (V(C_1) \cap R_i) = [|C_1|]$,
        \item for all $j \in [2,\ell]$, let $\pi_{R_i} (V(C_j) \cap R_i) = [\sum_{k=1}^{j-1} |V(C_k) \cap R_i| + 1, \sum_{k=1}^j |V(C_k) \cap R_i|]$,
    \end{itemize}
    where the ordering of the vertices of clique $C_j$ in $\pi_{R_i}$ is the same as in $\pi$,
    thus $v^{C_j,i}_{\max}$ denotes the rightmost vertex of $C_j$ in $\pi_{R_i}$.

    Next, we define an ordering of the vertices of $L_i$.
    Let $C_1, \ldots, C_r$ denote the cliques of $G-S'$ that have vertices in $L_i$,
    where $\minofset{\pi}{V(C_j) \cap L_i} < \minofset{\pi}{V(C_{j+1}) \cap L_i}$,
    for all $j \in [\ell - 1]$.
    Then, define $\pi_{L_i} \colon L_i \to [|L_i|]$ as follows:
    \begin{itemize}
        \item $\pi_{L_i} (V(C_1) \cap L_i) = [|C_1|]$,
        \item for all $j \in [2,r]$, let $\pi_{L_i} (V(C_j) \cap L_i) = [\sum_{k=1}^{j-1} |V(C_k) \cap L_i| + 1, \sum_{k=1}^j |V(C_k) \cap L_i|]$,
    \end{itemize}
    where the ordering of the vertices of clique $C_j$ in $\pi_{L_i}$ is the same as in $\pi$,
    thus $v^{C_j,i}_{\min}$ denotes the leftmost vertex of $C_j$ in $\pi_{L_i}$.

    Intuitively, ordering $\pi_{R_i}$ (respectively, $\pi_{L_i}$) is obtained by ``sliding'' all the
    non-rightmost (respectively, non-leftmost) vertices of cliques next to the rightmost (respectively, leftmost) vertex of the corresponding clique.

    Lastly, define ordering $\pi_{M_i} \colon M_i \to [|M_i|]$ such that for every clique having vertices in $M_i$,
    its vertices appear consecutively in $\pi_{M_i}$, for some arbitrary ordering of the cliques.
    
    Consider the ordering $\pi' \colon V(G) \to [|V(G)|]$, such that $\pi'(s) = \pi(s)$, for all $s \in S'$,
    while for every $i \in [0,k']$ and every $v \in B_i$,
    it holds that
    \begin{itemize}
        \item $\pi'(v) = i + \sum_{j=0}^{i-1} |B_j| + \pi_{R_i}(v)$, if $v \in R_i$,
        \item $\pi'(v) = i + \sum_{j=0}^{i-1} |B_j| + |R_i| + \pi_{M_i}(v)$, if $v \in M_i$,
        \item $\pi'(v) = i + \sum_{j=0}^{i-1} |B_j| + |R_i| + |M_i| + \pi_{L_i}(v)$, if $v \in L_i$.
    \end{itemize}
    Let $C_1$ and $C_2$ be cliques such that $V(C_1) \cap R_i \neq \varnothing$ while $V(C_2) \cap L_i \neq \varnothing$.
    Notice that $\maxofset{\pi'}{V(C_1) \cap B_i} = \pi' (v^{C_1,i}_{\max})$,
    while $\pi' (v^{C_1,i}_{\max}) \leq \pi(v^{C_1,i}_{\max})$,
    since for any vertex $v \in B_i$ such that $\pi'(v) < \pi' (v^{C_1,i}_{\max})$,
    it holds that $\pi (v) < \pi (v^{C_1,i}_{\max})$.
    In an analogous way,
    it follows that $\minofset{\pi'}{V(C_2) \cap B_i} = \pi' (v^{C_2,i}_{\min})$,
    while $\pi (v^{C_2,i}_{\min}) \leq \pi' (v^{C_2,i}_{\min})$,
    since for any vertex $v \in B_i$ such that $\pi' (v^{C_2,i}_{\min}) < \pi'(v)$,
    it holds that $\pi (v^{C_2,i}_{\min}) < \pi (v)$.

    It holds that $\pi'$ has Property~\hyperref[property:1]{$(\Pi_1)$}, it is compatible with $\sigma$ and $\mathcal{B}$,
    and since $\pi$ is $S'$-extremal, that is also the case for $\pi'$.
    It remains to argue that $\str(\pi') \leq \str(\pi)$.
    First, notice that since $\pi'$ is $S'$-extremal, it suffices to argue about the edges belonging to cliques of $G-S'$.
    Let $C$ be a cluster of distribution-type $\tau$, and let $e = \braces{v^{C,i_1}_{\min},v^{C,i_2}_{\max}} \in E(C)$ denote the edge of $E(C)$ of maximum stretch on ordering $\pi$,
    for some $i_1, i_2 \in [0,k']$.
    If all the vertices of $C$ are contained in a single bucket,
    then it holds that $i_1 = i_2 = i$, i.e., $V(C) \subseteq B_i$ for some $i \in [0,k']$.
    Since the vertices of $C$ appear consecutively in $\pi'$, it follows that $\stretchEdge{\pi'} (e) = |C| - 1$, which is minimum.
    Alternatively, notice that $e$ is the edge of $E(C)$ of maximum stretch also in $\pi'$,
    while $v^{C,i_1}_{\min} \in L_{i_1}$ and $v^{C,i_2}_{\max} \in R_{i_2}$.
    Then, $\stretchEdge{\pi'}(e) = \pi' (v^{C,i_2}_{\max}) - \pi' (v^{C,i_1}_{\min}) \leq \pi(v^{C,i_2}_{\max}) - \pi(v^{C,i_1}_{\min}) = \stretchEdge{\pi}(e)$ follows.
\end{proof}
\end{toappendix}

\begin{definition}[Property $(\Pi_2)$]\label{property:2}
    We say that an $S'$-extremal ordering $\pi \colon V(G) \to [n]$
    which is compatible with $\sigma$ and $\mathcal{B}$ has \emph{Property~\hyperref[property:2]{$(\Pi_2)$}} if
    for any two distribution-types $\tau, \tau' \in \distributionTypes$ and any two clusters $C$ and $C'$ of distribution-type $\tau$ and $\tau'$ respectively,
    the following holds.
    \begin{itemize} 
        \item If either $\leftBucket(\tau) = \leftBucket(\tau')$ or $\rightBucket(\tau) = \rightBucket(\tau')$,
        then for any $v \in V(C) \cap B_{\leftBucket(\tau)}$, $v' \in V(C') \cap B_{\leftBucket(\tau')}$,
        $w \in V(C) \cap B_{\rightBucket(\tau)}$, $w' \in V(C') \cap B_{\rightBucket(\tau')}$ 
        we have that $\pi(v) < \pi(v')$ if and only if $\pi(w) < \pi(w')$.
    \end{itemize}
\end{definition}

\begin{toappendix}
\begin{lemma}\label{lem:relativeOrder}
    Given an $S'$-extremal ordering $\pi \colon V(G) \to [n]$ which is compatible with $\sigma$ and $\mathcal{B}$,
    there exists an $S'$-extremal ordering $\pi' \colon V(G) \to [n]$ with $\str(\pi') \leq \str(\pi)$
    which is compatible with $\sigma$ and $\mathcal{B}$ and has Properties~\hyperref[property:1]{$(\Pi_1)$} and~\hyperref[property:2]{$(\Pi_2)$}.
\end{lemma}

\begin{proof}
    Let $\pi \colon V(G) \to [n]$ be an $S'$-extremal ordering of $V(G)$ which is compatible with $\sigma$ and $\mathcal{B}$.
    By \cref{lem:leftMiddleRightPartition_and_consecutiveCliques} we can assume that $\pi$ has Property~\hyperref[property:1]{$(\Pi_1)$}.
    For every clique $C$, let $v^{C}_{\min}$ and $v^{C}_{\max}$ denote its leftmost and rightmost vertex in ordering $\pi$,
    i.e., $\pi(v^{C}_{\min}) = \minofset{\pi}{V(C)}$ and $\pi(v^{C}_{\max}) = \maxofset{\pi}{V(C)}$.

    Fix an $\ell \in [0,k']$ and let $C$ and $C'$ be two clusters of distribution-type $\tau$ and $\tau'$ respectively, such that
    \begin{itemize}
        \item $\leftBucket(\tau) = \leftBucket(\tau') = \ell$,
        \item $\pi (V(C') \cap B_\ell) = [\maxofset{\pi}{V(C) \cap B_\ell} + 1, \maxofset{\pi}{V(C) \cap B_\ell} + |V(C') \cap B_\ell|]$,
        i.e., the vertices of $V(C') \cap B_\ell$ appear right after the vertices of $V(C) \cap B_\ell$ in the ordering $\pi$,
        \item $\pi (V(C') \cap B_{\rightBucket(\tau')}) < \pi (V(C) \cap B_{\rightBucket(\tau)})$.
    \end{itemize}
    Notice that it holds that $\rightBucket(\tau) \neq \ell$ and $\rightBucket(\tau') \neq \ell$,
    since otherwise $\pi$ does not have Property~\hyperref[property:1]{$(\Pi_1)$}.
    Let $\pi_{C\leftrightarrow C'} \colon V(G) \to [|V(G)|]$ be the ordering
    obtained by ``sliding'' the vertices of $V(C') \cap B_\ell$ to the left of the vertices of $V(C) \cap B_\ell$,
    i.e.,
    \begin{itemize}
        \item $\pi_{C\leftrightarrow C'}(v) = \pi(v)$, for all $v \notin (V(C) \cup V(C')) \cap B_\ell$,
        \item $\pi_{C\leftrightarrow C'}(v) = \pi(v) -\leftSize(\tau)$, for all $v \in V(C') \cap B_\ell$,
        \item $\pi_{C\leftrightarrow C'}(v) = \pi(v) +\leftSize(\tau')$, for all $v \in V(C) \cap B_\ell$,
    \end{itemize}
    where the ordering of the vertices in every clique $C$ remains as in $\pi$,
    thus for $v^C_{\min}$ and $v^C_{\max}$ it holds that $\pi_{C\leftrightarrow C'}(v^{C}_{\min}) = \minofset{\pi_{C\leftrightarrow C'}}{V(C)}$
    and $\pi_{C\leftrightarrow C'}(v^{C}_{\max}) = \maxofset{\pi_{C\leftrightarrow C'}}{V(C)}$.
    Notice that due to Property~\hyperref[property:1]{$(\Pi_1)$}, for all vertices $v$ for which $\pi_{C\leftrightarrow C'} (v) \neq \pi(v)$, it holds that $\pi(v) \in I^{\ell}_L$.
    It holds that $\pi_{C\leftrightarrow C'}$ has Property~\hyperref[property:1]{$(\Pi_1)$}, it is compatible with $\sigma$ and $\mathcal{B}$,
    and since $\pi$ is $S'$-extremal, that is also the case for $\pi_{C\leftrightarrow C'}$.
    We argue that $\str(\pi_{C\leftrightarrow C'}) \leq \str(\pi)$.
    Notice that it suffices to argue about the stretches of edges in $E(C) \cup E(C')$.
    Let $e = \braces{v^C_{\min}, v^C_{\max}}$ and $e' = \braces{v^{C'}_{\min}, v^{C'}_{\max}}$ denote the edge of maximum stretch (in both orderings $\pi$ and $\pi_{C\leftrightarrow C'}$)
    of $E(C)$ and $E(C')$ respectively.
    Since $\pi (v^C_{\min}) < \pi_{C\leftrightarrow C'} (v^C_{\min})$,
    while $\pi (v^C_{\max}) = \pi_{C\leftrightarrow C'} (v^C_{\max})$, it follows that $\stretchEdge{\pi_{C\leftrightarrow C'}} (e) < \stretchEdge{\pi} (e)$.
    As for $e'$, it holds that $\pi_{C\leftrightarrow C'} (v^{C'}_{\min}) = \pi (v^{C}_{\min})$, while $\pi_{C\leftrightarrow C'} (v^{C'}_{\max}) < \pi (v^{C}_{\max})$,
    therefore $\stretchEdge{\pi_{C\leftrightarrow C'}} (e') < \stretchEdge{\pi} (e)$ follows.

    In an analogous way, one can prove that for any pair of clusters $C$ and $C'$ of distribution-type $\tau$ and $\tau'$ respectively,
    where
    \begin{itemize}
        \item $\rightBucket(\tau) = \rightBucket(\tau') = r$ and
        \item the vertices of $V(C') \cap B_r$ appear right after the vertices of $V(C) \cap B_r$,
    \end{itemize}
    there exists an $S'$-extremal ordering that has Property~\hyperref[property:1]{$(\Pi_1)$}, it is compatible with $\sigma$ and $\mathcal{B}$,
    is of no bigger stretch, reorders only vertices $v$ for which $\pi (v) \in I^r_R$,
    and the vertices of $V(C) \cap B_{\leftBucket(\tau)}$ appear before those of $V(C') \cap B_{\leftBucket(\tau')}$.
    
    By exhaustively applying both arguments, and due to the transitivity of $<$, the statement follows.
\end{proof}
\end{toappendix}

\noindent
Lastly, we want the buckets to be ordered by distribution-types which will enable us to express the stretch within clusters by linear constraints.
To achieve this, we define two orderings of distribution-types,
dictating in which order (in a nice, optimal vertex ordering) cliques of a certain type will appear within a bucket.
First, let $\distributionTypes_R^i = \bigcup_{\genfrac{}{}{0pt}{}{\ell \in [0,i-1],}{n_L, n_R \in [\omega(G)]}} \distributionTypes^{(\ell,i,n_L,n_R)}$ and 
define the ordering
$\rho_i \colon \distributionTypes_R^i \to [|\distributionTypes_R^i|]$ in the following way.
For any $\tau \in \distributionTypes^{(\ell,i,n_L,n_R)}$, $\tau' \in \distributionTypes^{(\ell',i,n_L',n_R')}$, we have that $\rho_i(\tau) < \rho_i(\tau')$ if either
\begin{itemize}
    \item  $\ell < \ell'$ or 
    \item  $\ell = \ell'$, $n_L \geq n_R$ and $n_L' < n_R'$ or
    \item  $\ell = \ell'$, $n_L \geq n_R$, $n_L' \geq n_R'$ and $n_R < n_R'$ or
    \item  $\ell = \ell'$, $n_L < n_R$, $n_L' < n_R'$ and $n_L > n_L'$ or
    \item  $\ell = \ell'$, $n_L \geq n_R$, $n_L' \geq n_R'$, $n_R = n_R'$ and $\tau\leq_{\operatorname{lex}}\tau'$ or 
    \item  $\ell = \ell'$, $n_L < n_R$, $n_L' < n_R'$, $n_L = n_L'$ and $\tau\leq_{\operatorname{lex}}\tau'$. 
\end{itemize}

Here $\leq_{\operatorname{lex}}$ refers to the lexicographic order on matrices in $\distributionTypes$ where we read the entries by lines top to bottom. However, we can replace this by any total ordering ($\leq_{\operatorname{lex}}$ is an arbitrary choice).

Moreover, let $\distributionTypes_L^i=\bigcup_{\genfrac{}{}{0pt}{}{r \in [i+1,k'],}{n_L, n_R \in [\omega(G)]}} \distributionTypes^{(i,r,n_L,n_R)}$ 
and define the ordering $\lambda_i \colon \distributionTypes_L^i \to [|\distributionTypes_L^i|]$
by letting $\lambda_i(\tau) < \lambda_i(\tau')$ for any $\tau \in \distributionTypes^{(i,r,n_L,n_R)}$, $\tau' \in \distributionTypes^{(i,r',n_L',n_R')}$ if either
\begin{itemize}
    \item  $r < r'$ or 
    \item  $r = r'$ and $\rho_i(\tau) < \rho_i(\tau')$.
\end{itemize}

\begin{remark}\label{lem:computingTypeOrders}
    Note that we can compute all $\rho_i$ and $\lambda_i$ in time quadratic in the size of $\distributionTypes$. 
\end{remark}

\begin{definition}[Property $(\Pi_3)$]\label{property:3}
    We say that an $S'$-extremal ordering $\pi \colon V(G) \to [n]$ which is compatible with $\sigma$ and $\mathcal{B}$
    has \emph{Property~\hyperref[property:3]{$(\Pi_3)$}} if for every $i \in [0,k']$ we can partition the interval $\pi(B_i)$ into (possibly empty) intervals
    \[
       J_{R}^{(i,1)} < \dots < J_{R}^{(i,|\distributionTypes_R^i|)} < J_M^i < J_{L}^{(i,1)} < \dots < J_L^{(i,|\distributionTypes_L^i|)}
     \]
    such that for every distribution-type $\tau \in \distributionTypes$ and every cluster $C$ of type $\tau$ and every $j \in [|\distributionTypes_R^i|]$, $j' \in [|\distributionTypes_L^i|]$,
    \begin{itemize}
       \item $\pi(V(C) \cap B_i) \subseteq J_R^{(i,j)}$ if $\rho_i(\tau)=j$ and 
       \item $\pi(V(C) \cap B_i) \subseteq J_L^{(i,j')}$ if $\lambda_i(\tau)=j'$. 
    \end{itemize}
\end{definition}

\begin{lemmarep}\label{lem:finalOrder}
    Given an $S'$-extremal ordering $\pi \colon V(G) \to [n]$ which is compatible with $\sigma$ and $\mathcal{B}$,
    there exists an $S'$-extremal ordering $\pi' \colon V(G) \to [n]$ with $\str(\pi') \leq \str(\pi)$ which
    is compatible with $\sigma$ and $\mathcal{B}$ and
    has Properties~\hyperref[property:1]{$(\Pi_1)$},~\hyperref[property:2]{$(\Pi_2)$}, and~\hyperref[property:3]{$(\Pi_3)$}.
\end{lemmarep}

\begin{proof}
    Let $\pi \colon V(G) \to [n]$ be an $S'$-extremal ordering which is compatible with $\sigma$ and $\mathcal{B}$.
    By \cref{lem:relativeOrder} we can assume that $\pi$ has Properties~\hyperref[property:1]{$(\Pi_1)$} and~\hyperref[property:2]{$(\Pi_2)$}.
    For every $i \in [0,k']$ let $I_R^i < I_M^i < I_L^i$ be the partition of $\pi(B_i)$ as in Property~\hyperref[property:1]{$(\Pi_1)$}. 
    Since $\pi$ has Property~\hyperref[property:2]{$(\Pi_2)$} it holds that for any fixed $i \in [0,k']$ we can partition the interval $I_R^i$ into intervals $I_R^{(i,0)} < \dots < I_R^{(i,i-1)}$
    such that $\pi(V(C) \cap B_i) \subseteq I_R^{(i,j)}$ if $C$ is of distribution-type $\tau$, $\rightBucket(\tau)=i$, and $\leftBucket(\tau)=j$. 
    Equivalently, we can partition $I_L^i$ into intervals $I_L^{(i,i+1)} < \dots < I_L^{(i,k')}$ such that
    $\pi(V(C) \cap B_i) \subseteq I_L^{(i,j)}$ if $C$ is of distribution-type $\tau$, $\leftBucket(\tau)=i$, and $\rightBucket(\tau)=j$.
    
    Fix any $\ell < r \in [0,k']$ and let $\tau,\tau'$ be two distribution-types with $\leftBucket(\tau)=\leftBucket(\tau')=\ell$ and $\rightBucket(\tau)=\rightBucket(\tau')=r$.
    We further let $C$ be a cluster of distribution-type $\tau$ and $C'$ a cluster of distribution-type $\tau'$.
    We say that $C$ \emph{directly precedes} $C'$ if $\pi(V(C) \cap B_\ell) < \pi(V(C') \cap B_\ell)$ and $\pi(V(C) \cap B_\ell) \cup \pi(V(C') \cap B_\ell)$ is an interval.
    Note that if $\pi(V(C) \cap B_\ell) \cup \pi(V(C') \cap B_\ell)$ is an interval 
    then so is $\pi(V(C) \cap B_r) \cup \pi(V(C') \cap B_r)$, due to Property~\hyperref[property:2]{$(\Pi_2)$}.
    Moreover, note that $\pi(V(C) \cap B_i)$ is an interval for every clique $C$ and $i\in [0,k']$ as $\pi$ has Property~\hyperref[property:1]{$(\Pi_1)$}.
    If $C$ directly precedes $C'$ we define a new ordering  $\pi_{C\leftrightarrow C'} \colon V(G) \to [n]$ which corresponds to exchanging $C$ and $C'$ in $B_\ell$ and $B_r$ as follows.
    \begin{itemize}
        \item $\pi_{C\leftrightarrow C'}(v)=\pi(v)-\leftSize(\tau)$ for $v\in V(C') \cap B_\ell$,
        \item $\pi_{C\leftrightarrow C'}(v)=\pi(v)+\leftSize(\tau')$ for $v\in V(C) \cap B_\ell$,
        \item $\pi_{C\leftrightarrow C'}(v)=\pi(v)-\rightSize(\tau)$ for $v\in V(C') \cap B_r$,
        \item $\pi_{C\leftrightarrow C'}(v)=\pi(v)+\rightSize(\tau')$ for $v\in V(C) \cap B_r$, and
        \item $\pi_{C\leftrightarrow C'}(v)=\pi(v)$ for any vertex $v\notin (V(C)\cup V(C')) \cap (B_\ell\cup B_r)$.
    \end{itemize}
    \cref{claim:exchange} identifies when exchanging $C$ and $C'$ cannot increase the stretch.
    
    \begin{claim}\label{claim:exchange}
        If one of the following holds, then $\str(\pi_{C\leftrightarrow C'}) \leq \str(\pi)$.
        \begin{enumerate}
            \item $\leftSize(\tau) \leq \rightSize(\tau)$ and $\leftSize(\tau') \geq \rightSize(\tau')$.
            \item $\leftSize(\tau) \geq \rightSize(\tau)$, $\leftSize(\tau') \geq \rightSize(\tau')$, and $\rightSize(\tau) \geq \rightSize(\tau')$.
            \item $\leftSize(\tau) \leq \rightSize(\tau)$, $\leftSize(\tau') \leq \rightSize(\tau')$, and $\leftSize(\tau) \leq \leftSize(\tau')$.
        \end{enumerate}
    \end{claim}
    
    \begin{figure}[ht]
        \centering
        \centerline{ \includestandalone[width=1\textwidth]{figures/ordering_cliques_in_buckets}}
        \caption{Exchange arguments used in \cref{claim:exchange}.}
        \label{fig:exArgument}
    \end{figure}
    
    \begin{claimproof}
        First note that $\stretchEdge{\pi}(e)=\stretchEdge{\pi_{C\leftrightarrow C'}}(e)$ for any edge $e\notin E(C)\cup E(C')$.
        Let $v_{\min}^C$ (respectively, $v_{\max}^C$) denote the leftmost (respectively, rightmost) vertex of $V(C)$ with regard to $\pi_{C\leftrightarrow C'}$,
        that is, $\pi(v_{\min}^C) = \minofset{\pi_{C\leftrightarrow C'}}{V(C)}$ and
        $\pi(v_{\max}^C) = \maxofset{\pi_{C\leftrightarrow C'}}{V(C)}$.
        Define $v_{\min}^{C'}$ and $v_{\max}^{C'}$ analogously.
        Note that $\pi(v_{\min}^{C}) = \minofset{\pi}{V(C)}$,
        $\pi(v_{\max}^{C}) = \maxofset{\pi}{V(C)}$,
        $\pi(v_{\min}^{C'}) = \minofset{\pi}{V(C')}$,
        and $\pi(v_{\max}^{C'}) = \maxofset{\pi}{V(C')}$
        as the order of vertices within $C$ and $C'$ is the same in $\pi$ and in $\pi_{C\leftrightarrow C'}$.
        Note that $\stretchEdge{\pi_{C\leftrightarrow C'}}(e) \leq \stretchEdge{\pi_{C\leftrightarrow C'}}(\{v_{\min}^C,v_{\max}^C\})$ for any edge $e \in E(C)$
        and $\stretchEdge{\pi_{C\leftrightarrow C'}}(e') \leq \stretchEdge{\pi_{C\leftrightarrow C'}}(\{v_{\min}^{C'},v_{\max}^{C'}\})$ for any edge $e'\in E(C')$.
        Hence, we only consider the stretch of the two edges $\{v_{\min}^{C},v_{\max}^C\}$ and $\{v_{\min}^{C'},v_{\max}^{C'}\}$.
        Consider how the stretch of the edge   $\{v_{\min}^C,v_{\max}^C\}$ is affected by the exchange of $C$ and $C'$.
        Prior to the exchange, all $\#L(\tau')$ vertices of $V(C')\cap B_\ell$ contributed to the stretch of $\{v_{\min}^C,v_{\max}^C\}$
        while after the exchange they do not.
        On the other hand, after the exchange of $C$ and $C'$ all $\#R(\tau')$ vertices in $V(C')\cap B_r$ contribute to the stretch of $\{v_{\min}^C,v_{\max}^C\}$
        but prior to the exchange they did not.
        Other than that, the stretch of $\{v_{\min}^C,v_{\max}^C\}$ is not affected by the exchange (see \cref{fig:exArgument} for an illustration).
        Hence, the difference between the stretch of $\{v_{\min}^C,v_{\max}^C\}$ before and after the exchange is $\leftSize(\tau')-\rightSize(\tau')$.
        In a very similar way, we can determine the difference between the stretch of edges before and after the exchange of $C$ and $C'$ to obtain the following:
        \begin{romanenumerate}
            \item\label{dif:1} $\stretchEdge{\pi}(\{v_{\min}^C,v_{\max}^C\})-\stretchEdge{\pi_{C\leftrightarrow C'}}(\{v_{\min}^C,v_{\max}^C\})=\leftSize(\tau')-\rightSize(\tau')$,
            
            \item\label{dif:2} $\stretchEdge{\pi}(\{v_{\min}^{C'},v_{\max}^{C'}\})-\stretchEdge{\pi_{C\leftrightarrow C'}}(\{v_{\min}^{C'},v_{\max}^{C'}\})=\rightSize(\tau)-\leftSize(\tau)$,
            
            \item\label{dif:3} $\stretchEdge{\pi}(\{v_{\min}^C,v_{\max}^C\})-\stretchEdge{\pi_{C\leftrightarrow C'}}(\{v_{\min}^{C'},v_{\max}^{C'}\})=\rightSize(\tau)-\rightSize(\tau')$, and 
            
            \item\label{dif:4} $\stretchEdge{\pi}(\{v_{\min}^{C'},v_{\max}^{C'}\})-\stretchEdge{\pi_{C\leftrightarrow C'}}(\{v_{\min}^{C},v_{\max}^{C}\})=\leftSize(\tau')-\leftSize(\tau)$.
        \end{romanenumerate}
        
        Assuming $\leftSize(\tau)\leq \rightSize(\tau)$ and $\leftSize(\tau')\geq \rightSize(\tau')$,
        Equalities~(\ref{dif:1}) and~(\ref{dif:2}) imply that $\str(\pi_{C\leftrightarrow C'})\leq \str(\pi)$. 
        Assuming $\leftSize(\tau)\geq \rightSize(\tau)$, $\leftSize(\tau')\geq \rightSize(\tau')$, and $\rightSize(\tau)\geq \rightSize(\tau')$
        we obtain the statement using Equalities~(\ref{dif:1}) and~(\ref{dif:3}).
        Lastly, assuming $\leftSize(\tau)\leq \rightSize(\tau)$, $\leftSize(\tau')\leq \rightSize(\tau')$, and $\leftSize(\tau)\leq \leftSize(\tau')$
        we get the desired bound on the stretch using Equalities~(\ref{dif:2}) and~(\ref{dif:4}).
        For an illustration of the different cases see \cref{fig:exArgument}.
    \end{claimproof}
    
    We now successively exchange clusters $C$ of distribution-type $\tau$ with clusters $C'$ of distribution-type $\tau'$
    for any $\tau,\tau'\in \distributionTypes$ with $\leftBucket(\tau)=\leftBucket(\tau')=\ell$ and $\rightBucket(\tau)=\rightBucket(\tau')=r$
    whenever $C$ is directly preceding $C'$ and $\rho_\ell(\tau') < \rho_\ell(\tau)$.
    Call the resulting ordering $\pi'$ and note that in each step,
    the stretch does not increase by \cref{claim:exchange} and hence $\str(\pi') \leq \str(\pi)$.
    Furthermore, since any exchange maintains $\mathcal{B}$, the ordering $\pi'$ is compatible with $\sigma$ and $\mathcal{B}$.
    Even more so, the exchanges maintain the partition of intervals $I_R^{(i,0)} < \dots < I_R^{(i,i-1)} < I_M^i < I_L^{(i,i+1)} < \dots < I_L^{(i,k')}$.
    By our exchange strategy we obtain a partition of $I_R^{(i,\ell)}$ for every $i\in [0,k']$, $\ell\in [0,i-1]$ into intervals
    $J_R^{(i,\rho_i(\tau))}$ for every $\tau\in \bigcup_{n_L, n_R \in [\omega(G)]} \distributionTypes^{(\ell,i,n_L,n_R)}$ such that
    $\pi(V(C)\cap B_i)\subseteq J_R^{(i,\rho_{i}(\tau))}$ if $C$ is of type $\tau$.
    Similarly, for every $i\in [0,k']$, $r\in [i+1,k']$ we obtain a partition of $I_L^{(i,r)}$
    into intervals $J_L^{(i,\lambda_i(\tau))}$ for every $\tau\in\bigcup_{n_L, n_R \in [\omega(G)]} \distributionTypes^{(i,r,n_L,n_R)}$ such that $\pi(V(C)\cap B_i) \subseteq J_L^{(i,\lambda_{i}(\tau))}$ if $C$ is of type $\tau$.
    Furthermore, the exchange strategy guaranties that $J_R^{(i,\ell)}<J_R^{(i,\ell')}$  and $J_L^{(i,r)}<J_L^{(i,r')}$ for $\ell\leq \ell'$, $r\leq r'$. 
\end{proof}

\subsection{ILP formulation}
Throughout this subsection we let $G$ be a graph and $S$ a cluster deletion set of $G$ with $k = |S|$.
Furthermore, for the remainder of this subsection we fix $\mathcal{C}$ to be a representative set of clusters,
$S' = S \cup \bigcup_{C \in \mathcal{C}} V(C)$ the extended deletion set containing vertices from $\mathcal{C}$ and $S$,
and $k' = |S'|$. 

For every ordering $\sigma \colon S' \to [k']$, we will use an ILP to determine whether
there is an $S'$-extremal ordering $\pi \colon V(G) \to [n]$ of stretch at most $b$ which is compatible with $\sigma$.
The ILP  has two variables $x_\tau, y_\tau$ for every distribution-type $\tau \in \distributionTypes$.
The variable $x_\tau$ expresses how many clusters of $G-S'$ have distribution-type $\tau$ in an optimal $S'$-extremal ordering compatible with $\sigma$.
The variable $y_\tau$ is an indicator variable which is $1$ if and only if $x_\tau > 0$ and $0$ otherwise.
We further use $z_i$ for $i \in [0,k']$ in our ILP formulation as a placeholder for the expression
$\sum_{\tau \in \distributionTypes} (x_\tau \cdot \lVert \tau_i \rVert_1)$
which expresses the number of vertices in bucket $i$.
For an assignment $\alpha \colon \setdef{x_\tau, y_\tau}{\tau \in \distributionTypes} \to \mathbb{N}$ of the variables of our ILP,
we write $\alpha(z_i)$ to stand for the expression $\sum_{\tau \in \distributionTypes}(\alpha(x_\tau) \cdot \lVert \tau_i \rVert_1)$.

We further need the leftmost and rightmost neighbor of any vertex of $S$ in $S'$,
thus define $v_{\min,\sigma}^s, v_{\max,\sigma}^s \in S'$ such that
$\sigma (v_{\min,\sigma}^s) = \minofset{\sigma}{N(s)}$ and $\sigma (v_{\max,\sigma}^s) = \maxofset{\sigma}{N(s)}$,
for every ordering $\sigma \colon S' \to [k']$ and  $s \in S$.
Note that by choosing  $S$ to be minimum, we can assume that $S$ contains no vertex with no neighbors in $G-S$ and hence $v_{\min,\sigma}^s$ and $v_{\max,\sigma}^s$ are well-defined. 

For a fixed ordering $\sigma \colon S' \to [k']$, we can now formulate our set of linear constraints.
The first three constraints ensure that we choose the number of clusters that have a certain distribution-type in a feasible way.
That is, (\ref{cons:T1}) ensures that the quantities of distribution-types corresponding to an assignment of the variables $x_\tau$ correspond
to a valid choice of allocating each available cluster in the input graph $G$ a distribution-type.
As for (\ref{cons:T2}), it ensures that $v_{\min,\sigma}^s$ is indeed the leftmost neighbor of $s$
while $v_{\max,\sigma}^s$ is the rightmost neighbor of $s$ for every $s\in S$ by ensuring that
any distribution-type placing a neighbor of $s$ in a bucket left of $v_{\min,\sigma}^s$ or right of $v_{\max,\sigma}^s$ does not occur.
Finally, (\ref{cons:T3}) guarantees that $y_\tau$ indeed indicates whether or not distribution-type $\tau$ is used in the solution.

\begin{description}
    \item[(T1)\namedlabel{cons:T1}{T1}] For every $\kappa \in \clusterTypes$,
    \[\#\kappa = \min\{\#\kappa, 2k\} + \sum_{\tau \in \distributionTypes_\kappa} x_\tau.\]
    
    \item[(T2)\namedlabel{cons:T2}{T2}]\label{equ:T3} For every $s \in S$ and every $\tau \in \distributionTypes$ for which $\tau_{N,i} > 0$
    for some $N \ni s$ and $i \in [0, \sigma(v_{\min,\sigma}^s)-1] \cup [\sigma(v_{\max,\sigma}^s),k']$,
    \[x_\tau=0. \] 

    \item[(T3)\namedlabel{cons:T3}{T3}] $x_\tau \cdot (1-y_\tau) = 0$ and  $(1 - x_\tau) \cdot y_\tau \leq 0$ for every $\tau \in \distributionTypes$.
\end{description}
    
The purpose of all remaining constraints is to ensure that for the assignment of variables,
which essentially corresponds to choosing a bucket distribution $\mathcal{B}$,
there is an $S'$-extremal ordering $\pi \colon V(G) \to [n]$ which is compatible with $\sigma$ and $\mathcal{B}$ for which $\str(\pi)\leq b$.
(\ref{cons:DS}) expresses that the stretch of edges in $G[S']$ is bounded by $b$.

\begin{description}
\item[(DS)\namedlabel{cons:DS}{DS}] For every $s,s'\in S'$ with $\{s,s'\} \in E(G)$, $\sigma(s)<\sigma(s')$,
\[b \geq \sigma(s') - \sigma(s) + \sum_{i\in [\sigma(s), \sigma(s')-1]} z_i.\]


\end{description}

The last three constraints deal with bounding the stretch of edges within clusters.
For this we assume that the $S'$-extremal ordering which is consistent with $\sigma$ and $\mathcal{B}$ is nice,
i.e., has Properties~\hyperref[property:1]{$(\Pi_1)$},~\hyperref[property:2]{$(\Pi_2)$}, and~\hyperref[property:3]{$(\Pi_3)$}.
The first constraint (\ref{cons:C1}) is necessary to bound the stretch of clusters that are fully contained in one bucket.
To bound the stretch of clusters contained in multiple buckets, we have one constraint for every
distribution-type $\tau\in \distributionTypes^{(\ell,r,n_L,n_R)}$ for any $\ell < r \in [0,k']$, $n_L, n_R \in [\omega(G)]$.
By Property~\hyperref[property:3]{$(\Pi_3)$} we know that there are intervals $J_L^{(\ell, \lambda_\ell(\tau))}$
containing all vertices from $B_\ell \cap V(C)$ and $J_R^{(r, \rho_r(\tau))}$ containing all vertices $B_r \cap V(C)$ for every cluster $C$ of distribution-type $\tau$.
The trick now is to observe that if $n_L \geq n_R$ then the first cluster appearing in $J_L^{(\ell, \lambda_\ell(\tau))}$
observes the maximum stretch while if $n_L<n_R$ it is the last clique.
Using this we can express with constraints (\ref{cons:C2}) and (\ref{cons:C3}) that the stretch of every cluster of distribution-type $\tau$ is bounded by $b$.

\begin{description}
    \item[(C1)\namedlabel{cons:C1}{C1}] $b \geq \omega(G)-1$.
    
    \item[(C2)\namedlabel{cons:C2}{C2}] For every $\ell < r \in [0,k']$, $n_L \geq n_R \in [\omega(G)]$, and $\tau \in \distributionTypes^{(\ell,r,n_L,n_R)}$,   
    \begin{align*}
        b \geq y_\tau \cdot
        \bigg(
        \sum_{\tau'\in \lambda_{\ell}^{-1}\big([\lambda_\ell(\tau), |\distributionTypes^\ell_L|]\big)}
        \leftSize(\tau') \cdot x_{\tau'}+\sum_{\ell<i<r}z_i + (r-\ell) \\
        +\sum_{\tau'\in \rho_{r}^{-1}\big([1,\rho_r(\tau)-1]\big)}\rightSize(\tau')\cdot x_{\tau'}+
        n_R -1\bigg).
    \end{align*}
            
    \item[(C3)\namedlabel{cons:C3}{C3}] For every $\ell < r\in [0,k']$, $n_L < n_R \in [\omega(G)]$, and $\tau \in \distributionTypes^{(\ell,r,n_L,n_R)}$,
        \begin{align*}
             b \geq y_\tau \cdot \bigg(n_L+
            \sum_{\tau'\in \lambda_{\ell}^{-1}\big([\lambda_\ell(\tau)+1,|\distributionTypes^\ell_L|]\big)}\leftSize(\tau')\cdot x_{\tau'}+\sum_{\ell<i<r}z_i +
            (r-\ell) \\
            +\sum_{\tau'\in \rho_{r}^{-1}\big([1,\rho_r(\tau)]\big)}\rightSize(\tau')\cdot x_{\tau'}
             -1\bigg).
        \end{align*}


\end{description}

\begin{lemmarep}\label{lem:ILP_correctness}
    For any ordering $\sigma \colon S' \to [k']$,
    there is an $S'$-extremal ordering $\pi \colon V(G) \to [n]$ of stretch at most $b$ which is compatible with $\sigma$ if and only if
    the system of linear equation (\ref{cons:T1}, \ref{cons:T2}, \ref{cons:T3},  \ref{cons:DS},  \ref{cons:C1}, \ref{cons:C2}, \ref{cons:C3}) for $\sigma$ admits a solution.
\end{lemmarep}

\begin{proof} 
    Towards showing the forward direction, assume that there is an $S'$-extremal ordering $\pi \colon V(G) \to [n]$ which is compatible with $\sigma$ and $\str(\pi) \leq b$.
    Let $\mathcal{B} = (B_0, \ldots, B_{k'})$ be the bucket distribution associated to $\pi$,
    i.e., $B_0 = \setdef{v \in V(G)}{\pi(v) < \pi(\sigma^{-1}(1))}$,
    $B_{k'} = \setdef{v \in V(G)}{\pi(v) > \pi(\sigma^{-1}(k'))}$, and
    $B_i = \setdef{v \in V(G)}{\pi(\sigma^{-1}(i)) < \pi(v) < \pi(\sigma^{-1}(i+1))}$ for $i \in [k'-1]$.
    By \cref{lem:finalOrder} there is an $S'$-extremal ordering $\pi' \colon V(G) \to [n]$ with $\str(\pi') \leq \str(\pi) \leq b$ which
    is compatible with $\sigma$ and $\mathcal{B}$ and
    has Properties~\hyperref[property:1]{$(\Pi_1)$},~\hyperref[property:2]{$(\Pi_2)$}, and~\hyperref[property:3]{$(\Pi_3)$}.  
    
    We now describe how to obtain an assignment $\alpha \colon \setdef{x_\tau, y_\tau}{\tau\in \distributionTypes} \to \mathbb{N}$
    which yields a solution to the system of linear equations~(\ref{cons:T1},~\ref{cons:T2},~\ref{cons:T3},~\ref{cons:DS},~\ref{cons:C1},~\ref{cons:C2},~\ref{cons:C3}).
    For every $\tau \in \distributionTypes$ we set $\alpha(x_\tau)$ to be the number of clusters in $G-S'$ that have distribution-type
    $\tau$ in $\mathcal{B}$ while we set $\alpha(y_\tau) = 0$ if $\alpha(x_\tau) = 0$ and $\alpha(y_\tau) = 1$ otherwise.
    We will argue in the following that this yields the desired solution.

    Firstly,~(\ref{cons:T1}) is satisfied since $\min\{\#\kappa,2k\}$ is counting the clusters of cluster-type $\kappa$ that are in $\mathcal{C}$ while
    by choice of $\alpha(x_\tau)$ the sum $\sum_{\tau\in \distributionTypes_\kappa} \alpha(x_\tau)$ is counting
    all clusters of cluster-type $\kappa$ that are not contained in $\mathcal{C}$. 
    
    If (\ref{cons:T2}) was not satisfied then there is $s\in S$, $N\ni s$,  $i\in [0,\sigma(v_{\min,\sigma}^s)-1]\cup [\sigma(v_{\max,\sigma}^s),k']$,
    and $\tau\in \distributionTypes$ such that $\tau_{N,i}>0$ and $\alpha(x_\tau)>0$.
    Hence there is a cluster $C$ of distribution-type $\tau$.
    Because $\tau_{N,i}>0$ there is a vertex $c\in V(C)$ which is adjacent to $s$ such that $c \in B_i$.
    But since $i\in [0,\sigma(v_{\min,\sigma}^s)-1]\cup [\sigma(v_{\max,\sigma}^s),k']$ we get that either
    $\pi(c)<\minofset{\pi}{N(s)\cap S'}$ or $\pi(c)>\maxofset{\pi}{N(s)\cap S'}$,
    which is a contradiction since $\pi$ is $S'$-extremal.
    Hence~(\ref{cons:T2}) is satisfied.

    Furthermore, the constraint~(\ref{cons:T3}) is clearly satisfied.
    
    To show that~(\ref{cons:DS}) is satisfied first observe that $\pi'(s)=\sigma(s)+\sum_{i\in [0,\sigma(s)-1]} z_i$ for every $s\in S'$ where $z_i=|B_i|$,
    that is the position of $s$ in $\sigma$ plus the size of all the buckets in between elements of $S'$ that appear before $s$ in $\sigma$.
    Hence for any $s,s'\in S$ with $\{s,s'\}\in E(G)$ and $\sigma(s)<\sigma(s')$ it holds that
    $\stretchEdge{\pi'}(\{s,s'\})=|\sigma(s)-\sigma(s')|+\sum_{i\in [\sigma(s),\sigma(s')-1]}z_i$. Since $\str(\pi')\leq b$  this immediately implies (\ref{cons:DS}) is satisfied.
    
    Next observe that~(\ref{cons:C1}) is trivially satisfied as for any graph the bandwidth must be at least its clique number.

    To prove~(\ref{cons:C2}) is satisfied we use the nice ordering properties of $\pi'$.
    Let $J_{R}^{(i,1)} < \dots < J_{R}^{(i,|\distributionTypes_R^i|)} < J_M^i < J_{L}^{(i,1)} < \dots < J_{L}^{(i,|\distributionTypes_L^i|)}$ be
    the partition of $\pi'(B_i)$ for which we have Property~\hyperref[property:3]{$(\Pi_3)$}.
    Fix some $\ell < r \in [0,k']$, $n_L\leq n_R\in [\omega(G)]$, and  $\tau \in \distributionTypes^{(\ell,r,n_L,n_R)}$.
    If there is no cluster of distribution-type $\tau$ then $\alpha(x_\tau) = \alpha(y_\tau)=0$ and (\ref{cons:C2}) is satisfied.
    Hence assume that there is at least one cluster of distribution-type $\tau$.
    For every cluster $C$ of distribution-type $\tau$ we let $v_{\min}^C,v_{\max}^C\in V(C)$ be the vertices such that
    $\pi'(v_{\min}^C)\leq \pi'(v)\leq \pi'(v_{\max}^C)$ for every $v\in V(C)$.
    Note that $v_{\min}^C \in B_\ell$ and $v_{\max}^C \in B_r$ as $C$ has distribution-type $\tau$ and $\leftBucket(\tau)=\ell$, $\rightBucket(\tau)=r$.
    We let $C$ be the cluster of distribution-type $\tau$ for which $\pi'(v_{\min}^C)< \pi'(v_{\min}^{C'})$ for any other cluster $C'$ of distribution-type $\tau$.
    We now let $X$ be the set of all vertices $v\in V(G)$ such that  $\pi'(v_{\min}^C)\leq \pi'(v)\leq \pi'(v_{\max}^C)$.
    By choice of $C$ we know that $J^{(\ell, \lambda_\ell(\tau))}_L \subseteq X$ as this includes all vertices
    $V(C')\cap B_\ell$ for every cluster $C'$ of distribution-type $\tau$.
    Additionally, $J_L^{(\ell, j)}\subseteq X$ for all $j \in [\lambda_\ell (\tau) +1, |\distributionTypes^\ell_L|]$
    according to the ordering of intervals $J_L^{(\ell,j)}$.
    Since $|J_L^{(\ell,\lambda_\ell(\tau'))}|=\leftSize(\tau') x_{\tau'}$
    we conclude that $|X\cap B_\ell|\geq \sum_{\tau'\in \lambda_{\ell}^{-1}\big([\lambda_\ell(\tau), |\distributionTypes^\ell_L|]\big)}\leftSize(\tau')\cdot x_{\tau'}$.
    On the other hand, $X\cap B_r$ contains every interval $J_R^{(r, i)}$ for which $i \in [1, \rho_r(\tau)-1]$ since
    these are all intervals of bucket $r$ preceding $J_R^{(r,\rho_r(\tau))}$.
    Additionally, $X \cap B_r$ contains the $n_R$ vertices belonging to $V(C) \cap B_r$.
    Hence, $|X\cap B_r|\geq \sum_{\tau'\in \rho_{r}^{-1}\big([1,\rho_r(\tau)-1]\big)}\rightSize(\tau')\cdot x_{\tau'}+ n_R$.
    Finally, $X \setminus (B_\ell \cup B_r)$ consists of precisely the buckets and vertices from $S'$ residing in between $B_\ell$ and $B_r$.
    This further gives us $|X \setminus (B_\ell\cup B_r)| = \sum_{\ell < i < r} z_i + (r - \ell)$.
    Since $\braces{v_{\min}^C,v_{\max}^C} \in E(G)$ and $\stretchEdge{\pi'}(\{v_{\min}^C,v_{\max}^C\})=|X|-1\leq b$ the constraint (\ref{cons:C2}) is satisfied.
    
    Validity of constraint (\ref{cons:C3}) can be shown with a symmetric argument in which
    we choose $C$ to be the cluster with $\pi'(v_{\max}^C)>\pi'(v_{\max}^{C'})$ for any other cluster $C'$ of distribution-type $\tau$
    and then mirror the roles of the left and right bucket.
    This concludes the proof that we indeed find a solution of the system of linear equations
    (\ref{cons:T1}, \ref{cons:T2}, \ref{cons:T3},  \ref{cons:DS}, \ref{cons:C1}, \ref{cons:C2}, \ref{cons:C3}).

    \bigskip

    Towards showing the backwards direction, assume that the ILP admits a solution and let
    $\alpha \colon \setdef{x_\tau,y_\tau}{\tau\in \distributionTypes} \to \mathbb{N}$
    be an assignment to the variables satisfying all constraints.
    We define an ordering $\pi \colon V(G) \to [n]$ in the following way. 
    For every $\kappa\in \clusterTypes$ we pick arbitrarily a partition $(P_\tau)_{\tau\in \distributionTypes_\kappa}$ of the set
    of clusters of type $\kappa$ excluding clusters in $\mathcal{C}$ such that $|P_\tau|=\alpha(x_\tau)$.
    This is possible because (\ref{cons:T1}) is satisfied.
    We now pick a bucket distribution $\mathcal{B}=(B_0,\dots,B_{k'})$ of $S'$ as follows.
    For every $\tau\in \distributionTypes$ and every $C\in P_\tau$ we pick
    arbitrarily $\tau_{N,i}$ vertices $v$ of $C$ with $N_S(v)=N$ and add them to bucket $i$ for every $i \in [0, k']$.
    This is possible since any clique allocated to $P_\tau$ for some $\tau\in \distributionTypes_\kappa$ has
    cluster-type $\kappa$ and therefore contains precisely $\kappa_N = \sum_{i \in [0,k']}\tau_{N,i}$ vertices $v$ with $N_S(v)=N$ by definition of $\distributionTypes_\kappa$.
    Clearly, every $C\in P_\tau$ has distribution-type $\tau$ in the bucket-distribution $\mathcal{B}$.
    We now obtain $\pi$ by choosing an arbitrary ordering which is compatible with $\sigma$ and $\mathcal{B}$,
    i.e., we order the vertices of $S'$ as prescribed by $\sigma$ and
    place vertices in $B_i$ in between vertex $\sigma^{-1}(i)$ and $\sigma^{-1}(i+1)$ while ordering the vertices in every $B_i$ arbitrarily. 
    
    We first argue that the ordering $\pi$ is $S'$-extremal.
    Recall that, for $s \in S$, $v_{\min,\sigma}^s$ and $v_{\max,\sigma}^s$ denote the leftmost and rightmost neighbor of $s \in S$ in $S'$,
    i.e., $\sigma (v_{\min,\sigma}^s) = \minofset{\sigma}{N(s) \cap S'}$ and $\sigma (v_{\max,\sigma}^s) = \maxofset{\sigma}{N(s) \cap S'}$.    
    Assume that $\pi$ is not $S'$-extremal and let $v\in V(G)\setminus S'$ be a vertex such that for some $s\in S$
    either $\pi(v) = \minofset{\pi}{N(s)}$ or $\pi(v) = \maxofset{\pi}{N(s)}$,
    which implies that either $\pi(v) < \pi(v_{\min,\sigma}^s)$ or $\pi(v) > \pi(v_{\max,\sigma}^s)$ respectively. 
    Let $C$ be the cluster containing $v$, $\tau$ the distribution-type of $C$ in $\mathcal{B}$, and $i\in [0,k']$ for which $v\in B_i$.
    By definition the existence of $v$ and $C$ of distribution-type $\tau$ implies that
    $\tau_{N_S(v),i}>0$ and $\alpha(x_\tau)>0$.
    Since~(\ref{cons:T2}) is satisfied it follows that $i \in [\sigma(v_{\min,\sigma}^s),\sigma(v_{\max,\sigma}^s)-1]$.
    Consequently, $\pi(v_{\min,\sigma}^s)<\pi(v)<\pi(v_{\max,\sigma}^s)$ since $\pi$ is compatible with $\sigma$ and $\mathcal{B}$ which yields a contradiction.

    Note that the ordering $\pi$ does not necessarily have $\str(\pi)\leq b$.
    But by \cref{lem:finalOrder} we can obtain an $S'$-extremal ordering $\pi' \colon V(G) \to [n]$ which is compatible with
    $\sigma$ and $\mathcal{B}$ and has Properties~\hyperref[property:1]{$(\Pi_1)$},~\hyperref[property:2]{$(\Pi_2)$}, and~\hyperref[property:3]{$(\Pi_3)$}.
    We claim that $\pi'$ satisfies that $\str(\pi')\leq b$ and hence is the desired ordering.

    To argue that $\str(\pi')\leq b$ we first consider any edge $e=\{v,s\}$ for which $v\in V(G)$ and $s\in S$.
    Consider the case that $\pi'(v) < \pi'(s)$ (the other case can be argued in a symmetric way).
    Since $\pi'$ is $S'$-extremal and $v$ is adjacent to $s$ we know that $\pi'(v_{\min,\sigma}^s) \leq \pi'(v)$.
    Hence
    \[
        \stretchEdge{\pi'}(e)\leq
        \stretchEdge{\pi'}(\{v_{\min,\sigma}^s,s\})=
        |\sigma(s)-\sigma(v_{\min,\sigma}^s)|+\sum_{i\in [\sigma(v_{\min,\sigma}^s), \sigma(s)-1]}z_i\leq b,
    \] 
    where the last inequality is due to (\ref{cons:DS}).
    
    \medskip

    Now consider any edge $e=\{u, w\}\in E(G)$ such that $u,w \notin S'$.
    Since $S' \supseteq S$, where $S$ is a cluster deletion set,
    there is a cluster $C$ in $G-S'$ such that $u,w \in V(C)$.
    Assume $\tau \in \distributionTypes$ is the distribution-type of $C$ and additionally let $\ell=\leftBucket(\tau)$, $r=\rightBucket(\tau)$,
    $n_L=\leftSize(\tau)$, and $n_R=\rightSize(\tau)$.
    We let $v_{\min}^C, v_{\max}^C\in V(C)$ be the two vertices with $\pi'(v_{\min}^C) \leq \pi'(v) \leq \pi'(v_{\max}^C)$ for every $v\in V(C)$.
    Since $\stretchEdge{\pi'}(e) \leq \stretchEdge{\pi'}(\braces{v_{\min}^C,v_{\max}^C})$ it is sufficient to show that
    $\stretchEdge{\pi'}(\braces{v_{\min}^C,v_{\max}^C}) \leq b$.
    First assume that $\ell=r$, which implies that every vertex of $C$ is contained in the same bucket.
    Due to Property~\hyperref[property:1]{$(\Pi_1)$} it follows that all vertices of $C$ appear consecutively in $\pi'$,
    therefore $\stretchEdge{\pi'}(\braces{v_{\min}^C,v_{\max}^C}) \leq |C|-1 \leq \omega(G)-1 \leq b$, where the last inequality follows from~(\ref{cons:C1}).
    
    \medskip
    
    Now consider the case where $\ell < r$.
    We use Property~\hyperref[property:3]{$(\Pi_3)$} to obtain a bound on the stretch.
    Therefore, for every $i\in [0,k']$ let $J_{R}^{(i,1)} < \dots < J_{R}^{(i,|\distributionTypes^i_R|)} < J_M^i < J_{L}^{(i,1)} < \dots < J_L^{(i,|\distributionTypes^i_L|)}$
    be the partition of $\pi'(B_i)$ as in Property~\hyperref[property:3]{$(\Pi_3)$}.
    We let $X$ be the set of all vertices $v \in V(G)$ such that  $\pi'(v_{\min}^C)\leq \pi'(v)\leq \pi'(v_{\max}^C)$.  
    
    We first assume that $n_L\geq n_R$.
    Let $Y$ be the set of vertices $v$ such that $v\in V(C')$ for some cluster $C'$ of type $\tau$ and $\pi'(v)<\pi'(v_{\min}^C)$.
    Since $\pi'$ has Property~\hyperref[property:1]{$(\Pi_1)$} we know that there is a set of clusters $\mathcal{C}_Y$ such that $Y=\bigcup_{C'\in\mathcal{C}_Y}V(C')\cap B_\ell$.
    Since $\pi'$ has Property~\hyperref[property:2]{$(\Pi_2)$} we observe that $X\cap J_L^{(\ell,\lambda_\ell(\tau))} = J_L^{(\ell,\lambda_\ell(\tau))} \setminus \big( \bigcup_{C'\in \mathcal{C}_Y} (V(C')\cap B_\ell) \big)$ while
    $X\cap J_R^{(r,\rho_r(\tau))} = \bigcup_{C'\in \mathcal{C}_Y}(V(C')\cap B_r)\cup (V(C)\cap B_r)$.
    We further know that $X\setminus (J_L^{(\ell,\lambda_\ell(\tau))}\cup J_R^{(r,\rho_r(\tau))})$ consists of precisely the vertices in $J_L^{(\ell,i)}$ for every $i>\lambda_{\ell}(\tau)$,
    the vertices in $B_i$ for $\ell < i < r$,
    the vertices $\setdef{s \in S'}{\ell < \sigma(s) < r}$,
    as well as the vertices in $J_R^{(r,i)}$ for every $i<\rho_{r}(\tau)$.
    We conclude that
    \begin{align*}
        X
        = \Big(J_L^{(\ell,\lambda_\ell(\tau))} \setminus \bigcup_{C'\in \mathcal{C}_Y} (V(C') &\cap B_\ell) \Big) \cup
        \bigcup_{i \in [\lambda_{\ell}(\tau)+1, |\distributionTypes^i_L|]} J_L^{(\ell,i)}\\
        &\cup \bigcup_{\ell<i<r}B_i \cup \setdef{s\in S'}{\ell<\sigma(s)<r}\\
        &\cup \bigcup_{i\in [1,\rho_{r}(\tau)-1]}J_R^{(r,i)} \cup
        \bigcup_{C'\in \mathcal{C}_Y}\big(V(C')\cap B_r\big) \cup
        \big(V(C)\cap B_r\big).
    \end{align*}
    To determine the size of $X$ note that $J_L^{(\ell,\lambda_\ell(\tau'))}=\leftSize(\tau')\cdot \alpha(x_{\tau'})$ and $J_R^{(r,\rho_r(\tau'))}=\rightSize(\tau')\cdot \alpha(x_{\tau'})$. Hence we obtain
    \begin{align*}
        \stretchEdge{\pi'}(\{v_{\min}^C,v_{\max}^C\})&=|X|-1\\
        &= \sum_{\tau'\in \lambda_{\ell}^{-1}\big([\lambda_\ell(\tau), |\distributionTypes^i_L|]\big)}\leftSize(\tau')\cdot \alpha(x_{\tau'}) -
        \sum_{C'\in \mathcal{C}_Y}n_L\\
        &\qquad +\sum_{\ell<i<r}\alpha(z_i) +
                (r-\ell) \\
        &\qquad +\sum_{\tau'\in \rho_{r}^{-1}\big([1,\rho_r(\tau)-1]\big)}\rightSize(\tau')\cdot \alpha(x_{\tau'})+\sum_{C'\in \mathcal{C}_Y}n_R+n_R-1\\
        &\leq \sum_{\tau'\in \lambda_{\ell}^{-1}\big([\lambda_\ell(\tau),|\distributionTypes^i_L|]\big)}\leftSize(\tau')\cdot \alpha(x_{\tau'})
        +\sum_{\ell<i<r}\alpha(z_i) +
                (r-\ell) \\
                &\qquad +\sum_{\tau'\in \rho_{r}^{-1}\big([1,\rho_r(\tau)-1]\big)}\rightSize(\tau')\cdot \alpha(x_{\tau'})+
                n_R -1 \\
        &\leq b,
    \end{align*}
    where the first inequality follows due to the assumption that $n_L\geq n_R$
    (thus $\sum_{C'\in \mathcal{C}_Y} n_R - \sum_{C'\in \mathcal{C}_Y} n_L\leq 0$)
    and the last inequality follows due to~(\ref{cons:C2}) as $\alpha(y_\tau) = 1$
    (as there exists some cluster of distribution-type $\tau$).

    In the case that $n_L<n_R$ we can conclude that $\stretchEdge{\pi'}(\{v_{\min}^C,v_{\max}^C\})\leq b$ by a symmetric argument.
    Hence $\str(\pi') \leq b$ concluding the proof of the statement.
\end{proof}
    
Using \cref{lem:ILP_correctness} we obtain an FPT-algorithm
which computes $S$, $\#\kappa$ for every cluster-type $\kappa$, and arbitrary picks an extended deletion set $S'$.
Then, for every ordering $\sigma \colon S' \to [k']$, the algorithm verifies whether the ILP admits a solution in which case the input is a YES-instance of \BW.

\begin{toappendix}
\begin{proof}[Proof of \cref{thm:FPTAlgo}]
    Let $(G, b)$ be an instance of \BW,
    and let $S \subseteq V(G)$ be a cluster deletion set of $G$ of size $|S| = \clusterDel(G) = k$
    (such a set can be computed in time $1.811^k n^{\bO(1)}$ using the algorithm of~\cite{mst/Tsur21}).
    Notice that one can compute $\#\kappa$ for every cluster-type $\kappa \in \clusterTypes$ in polynomial time,
    by checking the $S$-neighborhood of every vertex in every clique of $G-S$.
    Moreover, the size of every clique of $G - S$ is at most $\omega(G)$.
    Let $S' = S \cup \bigcup_{C \in \mathcal{C}} C$ be an extended deletion set of $G$,
    for some representative clique set $\mathcal{C}$.
    Due to \cref{obs:sizeOfExtended,lem:computingTypeOrders},
    one can compute $\lambda_i$ and $\rho_i$ for $i \in [0, |S'|]$
    in time FPT in $\clusterDel(G) + \omega(G)$.
    Due to \cref{lem:stretchSClusters}, there exists an $S'$-extremal ordering $\pi$ of $V(G)$ such that $\str(\pi) = \bw(G)$.
    Notice that $\pi$ is $S'$-extremal and compatible with $\sigma$, for some ordering $\sigma$ of $S'$. 

    Fix an ordering $\sigma \colon S' \to [|S'|]$ of $S'$.
    Due to \cref{lem:ILP_correctness},
    one can verify whether there exists an $S'$-extremal ordering of $V(G)$ compatible with $\sigma$ of stretch at most $b$ by solving the ILP.
    Notice that there are exactly $|S'|!$ different orderings $\sigma$,
    thus one can decide whether $\bw(G) \le b$ by solving an ILP for each such ordering.
    
    Due to \cref{obs:sizeOfExtended} as well as the formulation of the ILP,
    it follows that both the number of ILPs solved as well as the number of variables of the ILPs are bounded
    by a function of $\clusterDel(G) + \omega(G)$.
    Since ILP is FPT parameterized by the number of variables, it follows that {\BW} is FPT parameterized by $\clusterDel(G) + \omega(G)$.
\end{proof}
\end{toappendix}

\begin{remark}
    Using a minimization ILP, we can in fact \emph{construct} an ordering of \emph{minimum} stretch
    (and not just argue about the existence of an ordering of stretch at most $b$),
    since all the exchange arguments of \cref{subsec:niceOrderings} are constructive.
\end{remark}

\section{W[1]-hardness parameterized by cluster vertex deletion number}\label{sec:parametrizedHardness}

In this section we prove that {\BW} is W[1]-hard when parameterized by the cluster vertex deletion number of the input graph.
In order to do so, we present a parameterized reduction from {\UBP} parameterized by the number of bins,
which is well-known to be W[1]-hard~\cite{jcss/JansenKMS13}.

\problemdef{\UBP}
{A multiset $A = \braces{a_1, \ldots, a_n}$ of integers in unary, as well as $k \in \N$.}
{Determine whether there is a partition of $[n]$ into $k$ subsets $\mathcal{I}_1, \ldots, \mathcal{I}_k$,
such that for all $i \in [k]$, $\sum_{j \in \mathcal{I}_i} a_j = \sum_{j \in [n]} a_j / k$.}

Before we present the details of the construction, we first give some high-level intuition.
For an instance $(A,k)$ of {\UBP} we want to construct an equivalent instance $(G, b)$ of \BW,
such that $\clusterDel(G) = f(k)$ for some function $f$.
Roughly, the graph $G$ consists of cliques representing the items of the {\UBP} instance and
cliques that act as delimiters separating the items contained in some bucket from the items contained in the next bucket.
However, in order to guarantee that the entirety of every item clique is placed in between two consecutive delimiter cliques
and that the values of the items in between two delimiter cliques add up to $B$ (the capacity of the bins in the {\UBP} instance $(A,k)$),
some extra structure is needed. 
First we introduce two cliques of size $b+1$ that will be used as boundaries.
By making each item clique and each delimiter clique of the graph adjacent to some vertex in both of the boundary cliques,
it follows that in any ordering of stretch at most $b$,
all item cliques and all delimiter cliques of the graph will be positioned in between the two boundary cliques.

As the size of the deletion set cannot depend on the number or values of the items,
item cliques cannot be incident to individual deletion set vertices.
This makes it tricky to enforce that every vertex of an item clique is contained in between
the same two delimiter cliques as a majority of the item cliques would not be incident to any
edge of maximum stretch and therefore allow them a lot of freedom of movement.
In order to cope with this issue, we introduce a perfect copy of the delimiter and item cliques,
as well as edges between the original cliques and their copies resulting in them becoming twice
as big consisting of a left part, the original vertices, and a right part, the copy vertices.
The left part of all cliques will be connected to the left boundary clique and will therefore appear to the left of the right parts.
The right part will be connected to the right hand boundary cliques.
The item cliques will now be kept in place by having maximum stretch between the vertices of the left part and the vertices of the right part.


\begin{theorem}\label{thm:hardness}
    {\BW} is W[1]-hard when parameterized by the cluster vertex deletion number of the input graph.
\end{theorem}

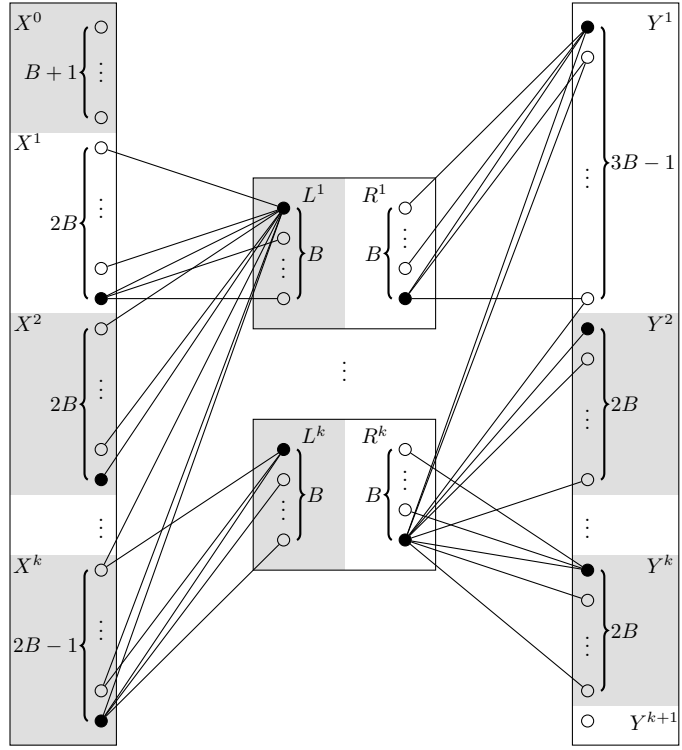
\begin{figure}[ht]
\centering
\begin{tikzpicture}[scale=0.8, transform shape]

\node[black_vertex] (c3C3) at (5,1) {};
\node[vertex] (c2C3) at (5,1.5) {};
\node[] () at (5,2.6) {$\vdots$};
\node[vertex] (c1C3) at (5,3.5) {};
\node[] () at (3.8,3.6) {$X^k$};
\draw [decorate,decoration = {brace,raise=5pt},thick] (5,1) --  (5,3.5) node[pos=0.5,left=0.25]{$2B-1$};

\node[] () at (5,4.3) {$\vdots$};

\node[black_vertex] (c3C2) at (5,5) {};
\node[vertex] (c2C2) at (5,5.5) {};
\node[] () at (5,6.6) {$\vdots$};
\node[vertex] (c1C2) at (5,7.5) {};
\node[] () at (3.8,7.6) {$X^2$};
\draw [decorate,decoration = {brace,raise=5pt},thick] (5,5) --  (5,7.5) node[pos=0.5,left=0.25]{$2B$};

\node[black_vertex] (c3C1) at (5,8) {};
\node[vertex] (c2C1) at (5,8.5) {};
\node[] () at (5,9.6) {$\vdots$};
\node[vertex] (c1C1) at (5,10.5) {};
\node[] () at (3.8,10.6) {$X^1$};
\draw [decorate,decoration = {brace,raise=5pt},thick] (5,8) --  (5,10.5) node[pos=0.5,left=0.25]{$2B$};

\node[vertex] (c3C0) at (5,11) {};
\node[] () at (5,11.85) {$\vdots$};
\node[vertex] (c1C0) at (5,12.5) {};
\node[] () at (3.8,12.6) {$X^0$};
\draw [decorate,decoration = {brace,raise=5pt},thick] (5,11) --  (5,12.5) node[pos=0.5,left=0.25]{$B+1$};

\begin{scope}[on background layer]
    \fill[lightgray!50] (3.5,0.6) rectangle (5.25,3.75);
    \fill[lightgray!50] (3.5,4.75) rectangle (5.25,7.75);
    \fill[lightgray!50] (3.5,10.75) rectangle (5.25,12.9);
\end{scope}
\draw[] (3.5,0.6) rectangle (5.25,12.9);

\node[vertex] (l3L1) at (8,8) {};
\node[] () at (8,8.6) {$\vdots$};
\node[vertex] (l2L1) at (8,9) {};
\node[black_vertex] (l1L1) at (8,9.5) {};
\draw [decorate,decoration = {brace,mirror,raise=5pt},thick] (8,8) --  (8,9.5) node[pos=0.5,right=0.25]{$B$};
\node[] () at (8.5,9.75) {$L^1$};

\draw[] (l1L1)--(c1C1);
\draw[] (l1L1)--(c2C1);
\draw[] (l1L1)--(c3C1);
\draw[] (l1L1)--(c1C2);
\draw[] (l1L1)--(c2C2);
\draw[] (l1L1)--(c3C2);
\draw[] (l1L1)--(c1C3);
\draw[] (l1L1)--(c2C3);
\draw[] (l1L1)--(c3C3);

\draw[] (c3C1)--(l2L1);
\draw[] (c3C1)--(l3L1);







\begin{scope}[shift={(0,3)}]
    
    \node[vertex] (l3L3) at (8,1) {};
    \node[] () at (8,1.6) {$\vdots$};
    \node[vertex] (l2L3) at (8,2) {};
    \node[black_vertex] (l1L3) at (8,2.5) {};
    \draw [decorate,decoration = {brace,mirror,raise=5pt},thick] (8,1) --  (8,2.5) node[pos=0.5,right=0.25]{$B$};
    \node[] () at (8.5,2.75) {$L^k$};
    
    \draw[] (l1L3)--(c1C3);
    \draw[] (l1L3)--(c2C3);
    \draw[] (l1L3)--(c3C3);
    
    \draw[] (c3C3)--(l2L3);
    \draw[] (c3C3)--(l3L3);
\end{scope}

\node[vertex] (Xk1) at (13,1) {};
\node[] () at (14.1,1) {$Y^{k+1}$};

\node[vertex] (x3X3) at (13,1.5) {};
\node[] () at (13,2.35) {$\vdots$};
\node[vertex] (x2X3) at (13,3) {};
\node[black_vertex] (x1X3) at (13,3.5) {};
\node[] () at (14.2,3.6) {$Y^k$};
\draw [decorate,decoration = {brace,mirror,raise=5pt},thick] (13,1.5) --  (13,3.5) node[pos=0.5,right=0.25]{$2B$};

\node[] () at (13,4.3) {$\vdots$};

\node[vertex] (x3X2) at (13,5) {};
\node[] () at (13,6.1) {$\vdots$};
\node[vertex] (x2X2) at (13,7) {};
\node[black_vertex] (x1X2) at (13,7.5) {};
\node[] () at (14.2,7.6) {$Y^2$};
\draw [decorate,decoration = {brace,mirror,raise=5pt},thick] (13,5) --  (13,7.5) node[pos=0.5,right=0.25]{$2B$};

\node[vertex] (x3X1) at (13,8) {};
\node[] () at (13,10.1) {$\vdots$};
\node[vertex] (x2X1) at (13,12) {};
\node[black_vertex] (x1X1) at (13,12.5) {};
\node[] () at (14.2,12.6) {$Y^1$};
\draw [decorate,decoration = {brace,mirror,raise=5pt},thick] (13,8) --  (13,12.5) node[pos=0.5,right=0.25]{$3B-1$};

\begin{scope}[on background layer]
    \fill[lightgray!50] (12.75,1.25) rectangle (14.5,3.75);
    \fill[lightgray!50] (12.75,4.75) rectangle (14.5,7.75);
\end{scope}
\draw[] (12.75,0.6) rectangle (14.5,12.9);

\node[black_vertex] (r3R1) at (10,8) {};
\node[vertex] (r2R1) at (10,8.5) {};
\node[] () at (10,9.1) {$\vdots$};
\node[vertex] (r1R1) at (10,9.5) {};
\draw [decorate,decoration = {brace,raise=5pt},thick] (10,8) --  (10,9.5) node[pos=0.5,left=0.25]{$B$};
\node[] () at (9.5,9.75) {$R^1$};
\draw[] (7.5,7.5) rectangle (10.5,10);
\begin{scope}[on background layer]
    \fill[lightgray!50] (7.5,7.5) rectangle (9,10);
\end{scope}

\draw[] (r3R1)--(x1X1);
\draw[] (r3R1)--(x2X1);
\draw[] (r3R1)--(x3X1);

\draw[] (x1X1)--(r1R1);
\draw[] (x1X1)--(r2R1);




\node[] () at (9,6.9) {$\vdots$};

\begin{scope}[shift={(0,3)}]
\node[black_vertex] (r3R3) at (10,1) {};
\node[vertex] (r2R3) at (10,1.5) {};
\node[] () at (10,2.1) {$\vdots$};
\node[vertex] (r1R3) at (10,2.5) {};
\draw [decorate,decoration = {brace,raise=5pt},thick] (10,1) --  (10,2.5) node[pos=0.5,left=0.25]{$B$};
\node[] () at (9.5,2.75) {$R^k$};
\draw[] (7.5,0.5) rectangle (10.5,3);
\begin{scope}[on background layer]
    \fill[lightgray!50] (7.5,0.5) rectangle (9,3);
\end{scope}

\draw[] (r3R3)--(x1X1);
\draw[] (r3R3)--(x2X1);
\draw[] (r3R3)--(x3X1);
\draw[] (r3R3)--(x1X2);
\draw[] (r3R3)--(x2X2);
\draw[] (r3R3)--(x3X2);
\draw[] (r3R3)--(x1X3);
\draw[] (r3R3)--(x2X3);
\draw[] (r3R3)--(x3X3);

\draw[] (x1X3)--(r1R3);
\draw[] (x1X3)--(r2R3);
\end{scope}

\end{tikzpicture}
\caption{Part of $G$, showing only the boundary and the delimiter cliques.
Rectangles denote cliques,
brackets denote number of vertices,
and black vertices compose a cluster deletion set.}
\label{fig:reduction_no_items}
\end{figure}

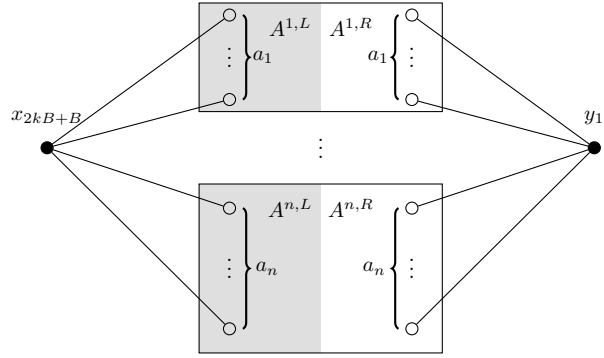
\begin{figure}[ht]
\centering
\begin{tikzpicture}[scale=0.8, transform shape]

\node[] () at (5,5.5) {$x_{2kB + B}$};
\node[black_vertex] (l) at (5,5) {};
\node[] () at (14,5.5) {$y_1$};
\node[black_vertex] (r) at (14,5) {};

\node[vertex] (i1l1) at (8,2) {};
\node[] () at (8,3.1) {$\vdots$};
\node[vertex] (i1l3) at (8,4) {};
\draw [decorate,decoration = {brace,mirror,raise=5pt},thick] (8,2) --  (8,4) node[pos=0.5,right=0.3]{$a_n$};
\draw[] (l)--(i1l1);
\draw[] (l)--(i1l3);

\node[] () at (9,4) {$A^{n,L}$};
\begin{scope}[on background layer]
    \fill[lightgray!50] (7.5,1.6) rectangle (9.5,4.4);
\end{scope}

\node[vertex] (i1r1) at (11,2) {};
\node[] () at (11,3.1) {$\vdots$};
\node[vertex] (i1r3) at (11,4) {};
\draw [decorate,decoration = {brace,raise=5pt},thick] (11,2) --  (11,4) node[pos=0.5,left=0.3]{$a_n$};
\draw[] (r)--(i1r1);
\draw[] (r)--(i1r3);

\node[] () at (10,4) {$A^{n,R}$};

\draw[] (7.5,1.6) rectangle (11.5,4.4);

\node[] () at (9.5,5.1) {$\vdots$};

\node[vertex] (i2l1) at (8,5.8) {};
\node[] () at (8,6.6) {$\vdots$};
\node[vertex] (i2l3) at (8,7.2) {};
\draw [decorate,decoration = {brace,mirror,raise=5pt},thick] (8,5.8) --  (8,7.2) node[pos=0.5,right=0.25]{$a_1$};
\draw[] (l)--(i2l1);
\draw[] (l)--(i2l3);

\node[] () at (9,7) {$A^{1,L}$};
\begin{scope}[on background layer]
    \fill[lightgray!50] (7.5,5.6) rectangle (9.5,7.4);
\end{scope}

\node[vertex] (i2r1) at (11,5.8) {};
\node[] () at (11,6.6) {$\vdots$};
\node[vertex] (i2r3) at (11,7.2) {};
\draw [decorate,decoration = {brace,raise=5pt},thick] (11,5.8) --  (11,7.2) node[pos=0.5,left=0.25]{$a_1$};
\draw[] (r)--(i2r1);
\draw[] (r)--(i2r3);
\node[] () at (10,7) {$A^{1,R}$};

\draw[] (7.5,5.6) rectangle (11.5,7.4);

\end{tikzpicture}
\caption{Rectangles denote cliques.
Black vertices compose a cluster deletion set.}
\label{fig:reduction_items}
\end{figure}

\subparagraph*{Construction.}
Let $(A,k)$ be an instance of \UBP, where $A = \braces{a_1, \dots, a_n}$.
Moreover, let $B = \sum_{j \in [n]} a_j / k$ be the capacity of every bin,
where $B \in \N$, since otherwise this would have been a trivial instance.
Set $b = 2kB+B-1$.
We will construct an equivalent instance $(G,b)$ of {\BW} as follows.

\proofsubparagraph*{Boundary cliques.}
First, we create two cliques $X$ and $Y$,
referred to as \emph{boundary cliques}, where
$V(X) = \braces{x_1, \ldots, x_{2kB+B}}$ and $V(Y) = \braces{y_1, \ldots, y_{2kB+B}}$.
We consider the following partition of the vertices of $X$:
let $X^0 = \braces{x_1, \ldots, x_{B+1}}$ and
for every $i \in [k-1]$ we denote the set $\braces{x_{2iB-B+2}, \ldots, x_{2iB+B+1}}$ by $X^i$,
while $X^k = \braces{x_{2kB - B + 2}, \ldots, x_{2kB + B}}$.
Note that $|X^0| = B+1$, $|X^k| = 2B-1$, and $|X^i| = 2B$, for all $i \in [k-1]$.
Moreover, we partition the vertices of $Y$ in a similar but slightly asymmetric way:
let $Y^1 = \braces{y_1, \ldots, y_{3B-1}}$ and
for every $i \in [2,k]$ we denote the set $\braces{y_{2iB-B}, \ldots, y_{2iB+B-1}}$ by $Y^i$,
while $Y^{k+1} = \braces{y_{2kB + B}}$.
Note that $|Y^1| = 3B - 1$, $|Y^{k+1}| = 1$, and $|Y^i| = 2B$, for all $i \in [2,k]$.

\proofsubparagraph*{Delimiter cliques.}
For every $i \in [k]$ we create a clique on vertex set $\{\ell_1^i,\dots,\ell_B^i, r_1^i,\dots,r_B^i\}$ of size $2B$.
We denote the set $\{\ell_1^i,\dots,\ell_B^i\}$ by $L^i$ and the set $\{r_1^i,\dots,r_B^i\}$ by $R^i$.
Moreover, let $L = \bigcup_{i=1}^k L^i$ and $R = \bigcup_{i=1}^k R^i$.
We add the following edges:
\begin{itemize}
    \item For every $i \in [k]$, $x \in \bigcup_{j = i}^k X^j$, we add the edge $\{\ell_1^i, x\}$.
    \item For every $i \in [k-1]$, $\ell \in L^i$, we add the edge $\braces{x_{2iB+B+1}, \ell}$.
    Moreover, we add an edge between $x_{2kB+B}$ and every vertex of $L^k$.
    \item For every $i\in [k]$, $y \in \bigcup_{j = 1}^i Y^j$, we add the edge $\{r_B^i, y\}$.
    \item For every $i \in [2,k]$, $r\in R^i$, we add the edge $\{y_{2iB-B}, r\}$. 
    Moreover, we add an edge between $y_1$ and every vertex of $R^1$.
\end{itemize}
For an illustration of the boundary and delimiter cliques, see \cref{fig:reduction_no_items}.

\proofsubparagraph*{Item cliques.}
For element $a_i \in A$, we  construct a  clique $A^i$ on vertex set $\setdef{a^{i,L}_j,a^{i,R}_j}{j \in [a_i]}$ of size $2a_i$.
We denote the set of vertices $\setdef{a^{i,L}_j}{j \in [a_i]}$ by $A^{i,L}$ and 
the set of vertices $\setdef{a^{i,R}_j}{j \in [a_i]}$ by $A^{i,R}$.
We add edges $\{x_{2kB+B},a \}$ for every $a \in \bigcup_{i\in [k]}A^{i,L}$ and
edges $\{y_1,a\}$ for every $a \in \bigcup_{i\in [k]}A^{i,R}$.
For an illustration, see \cref{fig:reduction_items}.

This concludes the construction of $G$.
\cref{fig:smallExample} illustrates an example of an ordering of stretch $b$ obtained by a YES-instance of \UBP.
In the following, we prove the equivalence of $(G,b)$ to the initial instance of \UBP.

\begin{figure}[ht]
    \centering
    \centerline{ \includegraphics[width=1.05\textwidth]{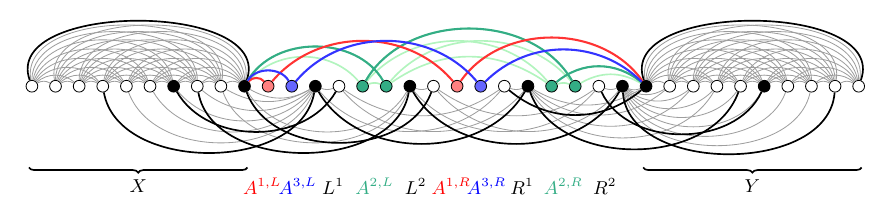}}
    \caption{For the instance $(\{a_1,a_2,a_3\}, 2)$ of {\UBP} with $a_1=1$, $a_2=2$, and $a_3=1$
    the figure shows the graph $G$ from the corresponding instance $(G,9)$ of \BW.
    Here the ordering of the vertices of $G$ with stretch $9$ corresponds to the solution of $(\{a_1,a_2,a_3\}, 2)$ in which $a_1,a_3$ are placed in the first bin and $a_2$ in the second.}
    \label{fig:smallExample}
\end{figure}

\begin{lemmarep}
    If $(A,k)$ is a YES-instance of \UBP, then $(G,b)$ is a YES-instance of \BW.
\end{lemmarep}

\begin{proof}
    Let $(\mathcal{I}_1, \ldots, \mathcal{I}_k)$ be a partition of $[n]$, where $\sum_{j \in \mathcal{I}_i} a_j = B$, for all $i \in [k]$.
    For all $i \in [k]$, set $\mathcal{S}_i = \bigcup_{j \in \mathcal{I}_i} A^{j,L}$ and 
    $\mathcal{S}'_i = \bigcup_{j \in \mathcal{I}_i} A^{j,R}$.
    Notice that $|\mathcal{S}_i| = |\mathcal{S}_i'| = B$,
    since $|A^{j,L}| = |A^{j,R}| = a_j$.
    
    For $i \in [k]$, we define $\pi_{\mathcal{S}_i} \colon \mathcal{S}_i \to [|\mathcal{S}_i|]$  to be an arbitrary ordering 
    of $\mathcal{S}_i$ and $\pi_{\mathcal{S}'_i} \colon \mathcal{S}'_i \to [|\mathcal{S}'_i|]$ 
    to be an arbitrary ordering of $\mathcal{S}'_i$.
    In order to define the ordering of the vertices of the graph $\pi \colon V(G) \to [|V(G)|]$,
    we will make use of these orderings.
    In particular, define $\pi$ as follows:
    \begin{itemize}
        \item First, let $\pi (x_i) = i$ for all $x_i \in V(X)$.
        \item For every $i \in [k]$, let
        \begin{itemize}
            \item $\pi (v) = 2kB + B + (i-1) 2B + \pi_{\mathcal{S}_i} (v)$ for every $v \in \mathcal{S}_i$,
            \item $\pi (\ell^i_j) = 2kB + 2B + (i-1) 2B + j = 2kB + 2iB + j$, for every $\ell^i_j \in L^i$, 
            \item $\pi (v) = 4kB + B + (i-1) 2B + \pi_{\mathcal{S}'_i} (v)$ for every $v \in \mathcal{S}'_i$,
            \item $\pi (r^i_j) = 4kB + 2B + (i-1) 2B + j = 4kB + 2iB + j$, for every $r^i_j \in R^i$. 
        \end{itemize}
        \item Lastly, let $\pi (y_i) = 6kB + B + i$ for all $y_i \in V(Y)$.
    \end{itemize}

    On a high level, we place first the vertices of $X$,
    then the vertices of $\mathcal{S}_1$, followed by the vertices of $L^1$,
    then of $\mathcal{S}_2$, and so on. 
    After the vertices of $L^k$, we proceed in an analogous way, by placing the vertices of $\mathcal{S}'_1$,
    followed by the vertices of $R^1$, then of $\mathcal{S}'_2$, and so on.
    After the vertices of $R^k$, we place the vertices of $Y$.
    The ordering of the vertices within $X, Y, L^i$, and $R^i$ is dictated by their index,
    whereas the ordering of $\mathcal{S}_i$ and $\mathcal{S}'_i$ is dictated by their corresponding ordering $\pi_{\mathcal{S}_i}$ and $\pi_{\mathcal{S}'_i}$.
    See \cref{fig:smallExample} for an illustration of the ordering by means of a small example instance of \UBP.
    
    It remains to show that $\stretchEdge{\pi}(e) \leq b = 2kB + B - 1$ for all $e \in E(G)$.
    In order to do so, we will consider all the possible cases regarding the endpoints of $e$.

    We first consider the edges $e \in E(X) \cup E(Y)$.
    Notice that the vertices of each of the cliques $X$ and $Y$ have been placed consecutively 
    in $\pi$, thus any edge $e \in E(X) \cup E(Y)$ has stretch at most equal to 
    the size of these cliques minus one, therefore $\stretchEdge{\pi}(e) \leq 2kB + B - 1 = b$.

    Now we consider the edges $e = \braces{x, \ell}$, where $x \in V(X)$ and $\ell \in L$.
    Notice that $\pi(x) < \pi(\ell)$.
    Suppose that $x \in X^i$ for $i \in [k]$,
    and set $x^i_{\min}$ and $x^i_{\max}$ such that $\pi (x^i_{\min}) = \minofset{\pi}{X^i}$ and $\pi(x^i_{\max}) = \maxofset{\pi}{X^i}$.
    Notice that it holds $x^i_{\min} = x_{2iB-B+2}$.
    Moreover $x^i_{\max} = x_{2iB+B+1}$ if $i \neq k$,
    otherwise $x^k_{\max} = x_{2kB+B}$, and $\pi(x^i_{\max}) \geq 2iB + B$ for all $i \in [k]$ follows.
    Notice that if $x \neq x^i_{\max}$, then $\maxofset{\pi}{N(x) \cap L} = \pi(\ell^i_1)$,
    since then $N(x) \cap L = \setdef{\ell^j_1}{j \in [i]}$.
    In that case, it holds that $\pi(\ell) - \pi(x) \leq \pi(\ell^i_1) - \pi(x^i_{\min}) = 2kB + 2iB + 1 - (2iB - B + 2) = 2kB + B - 1 = b$.
    Alternatively, if $x = x^i_{\max}$, then notice that $\maxofset{\pi}{N(x) \cap L} = \pi (\ell^i_B)$, since then $N(X) \cap L = \bigcup_{j=1}^i L^j$.
    In that case, it holds that $\pi(\ell) - \pi(x) \leq \pi(\ell^i_B) - \pi(x) \leq 2kB + 2iB + B - (2iB + B) = 2kB < b$.

    Next we consider the edges $e = \braces{y, r}$, where $y \in V(Y)$ and $r \in R$.
    Notice that $\pi(r) < \pi(y)$.
    Suppose that $y \in Y^i$ for $i \in [k]$,
    and set $y^i_{\max}$ and $y^i_{\min}$ such that $\pi (y^i_{\max}) = \maxofset{\pi}{Y^i}$ and $\pi(y^i_{\min}) = \minofset{\pi}{Y^i}$.
    Notice that it holds that $y^i_{\max} = y_{2iB+B-1}$.
    Moreover, $y^i_{\min} = y_{2iB-B}$ if $i \neq 1$,
    otherwise $y^1_{\min} = y_1$, and $\pi (y^i_{\min}) \leq 6kB + B + 2iB - B$ for all $i \in [k]$ follows.
    Notice that if $y \neq y^i_{\min}$, then $\minofset{\pi}{N(y) \cap R} = \pi(r^i_B)$,
    since $N(y) \cap R = \setdef{r^j_B}{j \in [i,k]}$.
    In that case, it holds that $\pi(y) - \pi(r) \leq \pi(y^i_{\max}) - \pi(r^i_B) = 6kB + B + 2iB + B - 1 - (4kB + 2iB + B) =2kB + B - 1 = b$.
    Alternatively, if $y = y^i_{\min}$, then notice that $\minofset{\pi}{N(y) \cap R} = \pi (r^i_1)$, since then $N(y) \cap R = \bigcup_{j=i}^k R^j$.
    In that case, it holds that $\pi(y) - \pi(r) \leq \pi(y) - \pi(r^i_1) \leq 6kB + B + 2iB - B - (4kB + 2iB + 1) = 2kB - 1 < b$.

    Now we consider the edges $e = \{ u,v \}$ where $ u,v \in L^i \cup R^i$ for some $i \in [k]$. 
    Notice that the single edge of maximum stretch out of those edges is $\braces{\ell^i_1, r^i_B}$,
    therefore $\stretchEdge{\pi}(e) \leq \pi(r^i_B) - \pi(\ell^i_1) = 4kB + 2iB + B - (2kB + 2iB + 1) = 2kB + B - 1 = b$.

    Lastly, consider the edges that are incident to the vertices of $A^j$ for some $j \in [n]$.
    We consider three cases.
    
    First, assume that both $u, v \in V(A^j)$.
    Then, notice that there exists $i \in [k]$ such that $A^{j,L} \subseteq \mathcal{S}_i$ and $A^{j,R} \subseteq \mathcal{S}'_j$.
    In that case, $\stretchEdge{\pi}(e) \leq \maxofset{\pi}{\mathcal{S}'_j} - \minofset{\pi}{\mathcal{S}_i}$.
    It holds that $\maxofset{\pi}{\mathcal{S}'_j} = \pi(r^i_1) - 1 = 4kB + 2iB$,
    while $\minofset{\pi}{\mathcal{S}_i} = \pi(\ell^{i-1}_B) + 1 = 2kB + 2iB - B + 1$ if $i \neq 1$ and $\minofset{\pi}{\mathcal{S}_1} = 2kB + B + 1$ otherwise, i.e., $\minofset{\pi}{\mathcal{S}_i} = 2kB + 2iB - B + 1$ for all $i \in [k]$.
    Consequently, it holds that $\stretchEdge{\pi}(e) \leq 4kB + 2iB - (2kB + 2iB - B + 1) = 2kB + B - 1 = b$.

    For the second case, assume that $e = \braces{u, x_{2kB+B}}$, where $u \in A^{j,L}$ for some $j \in [n]$.
    Notice that then $\pi(x_{2kB+B}) < \pi(u)$,
    while $\pi(u) \leq \pi(\ell^k_1) - 1 = 4kB$,
    therefore $\stretchEdge{\pi}(e) = \pi(u) - \pi(x_{2kB+B}) \leq 4kB - (2kB + B) = 2kB - B < b$.

    Finally, assume that $e = \braces{u, y_1}$, where $u \in A^{j,R}$ for some $j \in [n]$.
    Notice that then $\pi(u) < \pi(y_1)$,
    while $\pi(u) \geq \pi(\ell^k_B) + 1 = 4kB + B + 1$,
    therefore $\stretchEdge{\pi}(e) = \pi(y_1) - \pi(u) \leq 6kB + B + 1 - (4kB + B + 1) = 2kB < b$.
    
    Since it holds that $\stretchEdge{\pi}(e) \leq b$ for all $e \in E(G)$,
    it follows that $(G, b)$ is a YES-instance of \BW.
\end{proof}

\begin{lemmarep}
    If $(G,b)$ is a YES-instance of \BW, then $(A,k)$ is a YES-instance of \UBP.
\end{lemmarep}

\begin{proof}
    Let $\pi \colon V(G) \to [|V(G)|]$ be an ordering of the vertices of $G$ such that $\stretchEdge{\pi}(e) \leq b$,
    for all $e \in E(G)$.
    We first prove that the vertices of $X$ and $Y$ appear at the start and the end of the ordering respectively, or vice versa.

    \begin{claim}\label{claim:boundaries}
        One of the following holds:
        \begin{itemize}
            \item either $\pi (V(X)) = [b+1]$ and $\pi (V(Y)) = [|V(G)| - b, |V(G)|]$,
            \item or $\pi (V(Y)) = [b+1]$ and $\pi (V(X)) = [|V(G)| - b, |V(G)|]$.
        \end{itemize}
    \end{claim}

    \begin{claimproof}
        First, observe that since $X$ (respectively, $Y$) is a clique of size $b+1$,
        if its vertices were not occupying $b+1$ consecutive positions in $\pi$,
        there would exist an edge of stretch more than $b$, which is a contradiction.
        Consequently, the vertices of $X$ (respectively, $Y$) occupy $b+1$ consecutive positions in $\pi$.

        Now, assume that the vertices of $Z \in \braces{X, Y}$ appear neither at the start nor at the end of the ordering $\pi$.
        Then, there exist vertices placed in the ordering $\pi$ both before and after the positions occupied by the vertices of $Z$.
        Notice that $G' = G - V(Z)$ is connected, and let $v_1, v_2 \in V(G')$ such that $\braces{v_1, v_2} \in E(G')$
        and $\pi(v_1) < \pi(z) < \pi(v_2)$, for every $z \in V(Z)$.
        Since $|V(Z)| = b+1$, it follows that $\pi(v_2) - \pi(v_1) > b$, which is a contradiction.
    \end{claimproof}

    As a consequence of \cref{claim:boundaries}, it follows that the remaining $|V(G)| - 2 (b+1) = 2kB + 2 \sum_{i=1}^n a_i = 4kB$
    vertices of the delimiter and item cliques are positioned between the vertices of the boundary cliques in $\pi$.
    Without loss of generality, assume that the first case of \cref{claim:boundaries} holds,
    thus the vertices of $X$ appear at the start and the vertices of $Y$ at the end of ordering $\pi$
    (if that is not the case, all arguments made in the following for $\pi$ hold for the reverse ordering
    $\pi' (x) = |V(G)| + 1 - \pi(x)$ and hence the statement follows using a symmetric argument).
    Next, we will prove that the vertices of $X^0, X^1, X^2, \ldots, X^k$ and of $Y^1, \ldots, Y^k, Y^{k+1}$ appear sequentially in $\pi$,
    as depicted in \cref{fig:boundary_vertices_ordering}.
    Moreover, we will specify the positions of vertices $\ell^i_1$ and $r^i_B$, for $i \in [k]$.

    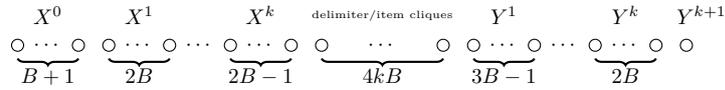
\begin{figure}[ht]
\centering
\begin{tikzpicture}[scale=0.8, transform shape]

\node[vertex] () at (3,5) {};
\node[] () at (3.5,5) {$\cdots$};
\node[] () at (3.5,5.5) {$X^0$};
\node[vertex] () at (4,5) {};
\draw [decorate,decoration = {brace,mirror,raise=5pt},thick] (3,5) --  (4,5) node[pos=0.5,below=0.25]{$B + 1$};

\node[vertex] () at (4.5,5) {};
\node[] () at (5,5) {$\cdots$};
\node[] () at (5,5.5) {$X^1$};
\node[vertex] () at (5.5,5) {};
\draw [decorate,decoration = {brace,mirror,raise=5pt},thick] (4.5,5) --  (5.5,5) node[pos=0.5,below=0.25]{$2B$};

\node[] () at (6,5) {$\cdots$};

\begin{scope}[shift={(2,0)}]
    \node[vertex] () at (4.5,5) {};
    \node[] () at (5,5) {$\cdots$};
    \node[] () at (5,5.5) {$X^k$};
    \node[vertex] () at (5.5,5) {};
    \draw [decorate,decoration = {brace,mirror,raise=5pt},thick] (4.5,5) --  (5.5,5) node[pos=0.5,below=0.25]{$2B-1$};
\end{scope}

\node[vertex] () at (8,5) {};
\node[] () at (9,5) {$\cdots$};
\node[] () at (9,5.5) {\tiny delimiter/item cliques};
\node[vertex] () at (10,5) {};
\draw [decorate,decoration = {brace,mirror,raise=5pt},thick] (8,5) --  (10,5) node[pos=0.5,below=0.25]{$4kB$};

\begin{scope}[shift={(6,0)}]
    \node[vertex] () at (4.5,5) {};
    \node[] () at (5,5) {$\cdots$};
    \node[] () at (5,5.5) {$Y^1$};
    \node[vertex] () at (5.5,5) {};
    \draw [decorate,decoration = {brace,mirror,raise=5pt},thick] (4.5,5) --  (5.5,5) node[pos=0.5,below=0.25]{$3B-1$};
\end{scope}

\node[] () at (12,5) {$\cdots$};

\begin{scope}[shift={(8,0)}]
    \node[vertex] () at (4.5,5) {};
    \node[] () at (5,5) {$\cdots$};
    \node[] () at (5,5.5) {$Y^k$};
    \node[vertex] () at (5.5,5) {};
    \draw [decorate,decoration = {brace,mirror,raise=5pt},thick] (4.5,5) --  (5.5,5) node[pos=0.5,below=0.25]{$2B$};
\end{scope}

\begin{scope}[shift={(11.5,0)}]
    \node[vertex] () at (2.5,5) {};
    \node[] () at (2.75,5.5) {$Y^{k+1}$};
\end{scope}

\end{tikzpicture}
\caption{Layout of $X^0, \ldots, X^k, Y^1, \ldots, Y^{k+1}$ in $\pi$, as derived in~\cref{claim:structure1}.}
\label{fig:boundary_vertices_ordering}
\end{figure}

    \begin{claim}\label{claim:structure1}
        Regarding the vertices of $X$ and $Y$,
        the following hold:
        \begin{itemize}
            \item $\pi(X^0) = [B+1]$,
            \item for every $i \in [k]$, $\pi(X^i) = [w^X_i, w^X_i + |X^i| - 1]$, where $w^X_i = B+1 + (i-1) 2B + 1 = 2iB - B + 2$, 
            \item $\pi(Y^1) = [6kB + B + 1, 6kB + 4B - 1]$,
            \item for every $i \in [2,k+1]$, $\pi(Y^i) = [w^Y_i, w^Y_i + |Y^i| - 1]$, where $w^Y_i = 6kB + 4B - 1 + (i-2) 2B + 1 = 6kB + 2iB$.
        \end{itemize}
        Additionally, for all $i \in [k]$, it holds that
        \begin{itemize}
            \item $\pi (\ell^i_1) = \minofset{\pi}{L^i} = 2kB + B + B + 2B \cdot (i-1) + 1 = 2kB + 2iB + 1$ 
            and
            \item $\pi (r^i_B) = \maxofset{\pi}{R^i} = 2kB + B + 2kB + i \cdot 2B = 4kB + 2iB + B$.
        \end{itemize}
    \end{claim}

    \begin{claimproof}
        Let for $v_1, v_2 \in V(G)$, $g_\pi (v_1, v_2) = \abs{\pi(v_1) - \pi(v_2)} + 1$ denote the number of vertices
        positioned between $\pi(v_1)$ and $\pi(v_2)$ including $v_1$ and $v_2$.
        Notice that if $e = \braces{v_1, v_2} \in E(G)$, then $g_\pi (v_1, v_2) \leq b+1$, since $\stretchEdge{\pi}(e) \leq b$.

        Fix some $i \in [k]$.
        Moreover, let $x^i_{\min}$ be the vertex such that $\pi (x^i_{\min}) = \minofset{\pi}{N(\ell^i_1) \cap V(X)}$,
        i.e., $x^i_{\min}$ denotes the leftmost positioned neighbor of $\ell^i_1$.
        Analogously, let $y^i_{\max}$ be the vertex such that $\pi (y^i_{\max}) = \maxofset{\pi}{N(r^i_B) \cap V(Y)}$,
        i.e., $y^i_{\max}$ denotes the rightmost positioned neighbor of $r^i_B$. 
        
        Notice that for $\ell^i_1$ it holds that $N(\ell^i_1) \cap V(X) = \bigcup_{j = i}^k X^j$,
        therefore $\ell^i_1$ has exactly $(k-i+1) 2B - 1$ neighbors belonging to $X$.
        On the other hand, for $r^i_B$ it holds that $N(r^i_B) \cap V(Y) = \bigcup_{j = 1}^i Y^j$,
        therefore $r^i_B$ has exactly $3B - 1 + (i-1) 2B = 2iB + B - 1$ neighbors belonging to $Y$.
        Consequently, it holds that $g_\pi (x^i_{\min}, y^i_{\max}) \geq (k-i+1) 2B - 1 + 2iB + B - 1 + 4kB = 6kB + 3B - 2 = 3(b+1) - 2$.
        Furthermore, $g_\pi (x^i_{\min}, y^i_{\max}) \leq g_\pi (x^i_{\min}, \ell^i_1) + g_\pi (\ell^i_1, r^i_B) + g_\pi (r^i_B, y^i_{\max}) - 2$,
        since at least two vertices, namely $\ell^i_1$ and $r^i_B$, are counted more than once.
        Since the edges $e^i_1 = \braces{x^i_{\min}, \ell^i_1}$, $e^i_2 = \braces{\ell^i_1, r^i_B}$ and $e^i_3 = \braces{r^i_B, y^i_{\max}}$
        are present in $G$, it follows that $g_\pi (x^i_{\min}, \ell^i_1) + g_\pi (\ell^i_1, r^i_B) + g_\pi (r^i_B, y^i_{\max}) - 2 \leq 3 (b+1) - 2$.
        Hence, it follows that $g_\pi (x^i_{\min}, \ell^i_1) = g_\pi (\ell^i_1, r^i_B) = g_\pi (r^i_B, y^i_{\max}) = b+1$,
        while $g_\pi (x^i_{\min}, y^i_{\max}) = 3 (b+1) - 2$,
        thus only vertices $\ell^i_1$ and $r^i_B$ are counted twice in the previous sum.
      
        Consequently, it holds that $\pi (\bigcup_{j=i}^k X^j) = [2kB + B - |\bigcup_{j=i}^k X^j| + 1, 2kB + B]$,
        while $\pi (\bigcup_{j=1}^i Y^i) = [6kB + B + 1, 6kB + B + 1 + |\bigcup_{j=1}^i Y^i| - 1]$,
        i.e., the vertices of $N(\ell^i_1) \cap V(X)$ appear in the rightmost positions occupied by the vertices of $X$,
        while the vertices of $N(r^i_B) \cap V(Y)$ appear in the leftmost positions occupied by the vertices of $Y$.
        Having determined $\pi (x^i_{\min})$ and $\pi(y^i_{\max})$, as well as that $\pi (\ell^i_1) - \pi (x^i_{\min}) = \pi (y^i_{\max}) - \pi (r^i_B) = b$,
        the values of $\pi (\ell^i_1)$ and $\pi (r^i_B)$ follow.
        Since $\pi (r^i_B) - \pi (\ell^i_1) = b$ and the vertices of $L^i \cup R^i$ form a clique,
        it follows that $\pi (\ell^i_1) = \minofset{\pi}{L^i}$ and $\pi (r^i_B) = \maxofset{\pi}{R^i}$;
        if that were not the case, then there would exist an edge of stretch larger than $b$.

        Lastly, since $N(\ell^i_1) \cap V(X) = (N(\ell^{i+1}_1) \cap V(X)) \cup X^i$,
        and $N(r^{i+1}_B) \cap V(Y) = (N(r^i_B) \cap V(Y)) \cup Y^{i+1}$,
        the vertices of $X^i$ and $Y^i$ are positioned as in \cref{fig:boundary_vertices_ordering}.
    \end{claimproof}

    Next, we prove that all vertices of $L^i$ (respectively, $R^i$) appear before the vertices of $L^{i+1}$ (respectively, $R^{i+1}$).

    \begin{claim}\label{claim:structure2}
        For $i \in [k]$ and $j \in [B]$, it holds that
        \begin{itemize}
            \item $\pi (\ell^i_1) \leq \pi (\ell^i_j) < \pi (\ell^i_1) + |X^i|$ and
            \item $\pi (r^i_B) - |Y^i| < \pi (r^i_j) \leq \pi (r^i_B)$.
        \end{itemize}
    \end{claim}

    \begin{claimproof}
        Fix some $i \in [k]$.
        Set $x = x_{2iB+B+1}$ if $i \neq k$ and $x = x_{2kB+B}$ otherwise.
        Additionally, set $y = y_{2iB-B}$ if $i \neq 1$ and $y = y_1$ otherwise.
        Notice that $x \in X^i$ while $y \in Y^i$.
        Let $x^i_{\min}$ and $y^i_{\max}$ be vertices such that $\pi (x^i_{\min}) = \minofset{\pi}{N(\ell^i_1)}$ and $\pi (y^i_{\max}) = \maxofset{\pi}{N(r^i_B)}$.
        Due to \cref{claim:structure1}, it holds that $\pi (x^i_{\min}) \leq \pi (x) \leq \pi (x^i_{\min}) + |X^i| - 1$,
        while $\pi (y^i_{\max}) - |Y^i| + 1 \leq \pi (y) \leq \pi (y^i_{\max})$.

        Let $j \in [B]$.
        We will prove that $\pi (\ell^i_j) < \pi (\ell^i_1) + |X^i|$ as well as $\pi (r^i_B) - |Y^i| < \pi (r^i_j)$,
        since the other inequalities follow from \cref{claim:structure1}.

        For the first inequality, notice that $\braces{x, \ell^i_j} \in E(G)$,
        therefore it holds that $\pi(\ell^i_j) - \pi (x) \leq b$.
        Moreover, due to \cref{claim:structure1} it holds that $\pi(\ell^i_1) - \pi(x^i_{\min}) = b$.
        In that case, it follows that
        \begin{align*}
            \pi (\ell^i_j) &\leq b + \pi (x)\\
            &= \pi(\ell^i_1) - \pi(x^i_{\min}) + \pi (x)\\
            &\leq \pi(\ell^i_1) - \pi(x^i_{\min}) + \pi (x^i_{\min}) + |X^i| - 1\\
            &= \pi(\ell^i_1) + |X^i| - 1.
        \end{align*}

        As for the second inequality, notice that $\braces{y, r^i_j} \in E(G)$,
        therefore it holds that $\pi(y) - \pi(r^i_j) \leq b$.
        Moreover, due to \cref{claim:structure1} it holds that $\pi(y^i_{\max}) - \pi(r^i_B) = b$.
        In that case,
        \begin{align*}
            \pi (r^i_j) &\geq \pi (y) - b\\
            &= \pi (y) - \pi(y^i_{\max}) + \pi(r^i_B)\\
            &\geq \pi (y^i_{\max}) - |Y^i| + 1 - \pi(y^i_{\max}) + \pi(r^i_B)\\
            &= \pi(r^i_B) - |Y^i| + 1,
        \end{align*}
        and the statement follows.
    \end{claimproof}

    In the following, we will refer to the vertex placed at position $2kB+B$ as $\ell^0_1$,
    while we will additionally refer to $\ell^k_1$ as $r^0_B$.
    Notice that, due to Claims~\ref{claim:structure1} and~\ref{claim:structure2},
    it follows that the vertices of the item cliques are positioned in the following way:
    \begin{itemize}
        \item $B$ vertices are placed between positions $\pi (\ell^0_1)$ and $\pi (\ell^1_1)$,
        \item $B$ vertices are placed between positions $\pi (\ell^i_1)$ and $\pi (\ell^{i+1}_1)$,
        for every $i \in [k-1]$,
        since $B-1$ out of the $2B - 1$ free positions are occupied by vertices of $L^i$,
        \item $B$ vertices are placed between positions $\pi (\ell^k_1)$ and $\pi (r^1_B)$,
        since $2(B-1)$ out of the $3B - 2$ free positions are occupied by vertices of $L^k$ and $R^1$,
        \item $B$ vertices are placed between positions $\pi (r^i_B)$ and $\pi (r^{i+1}_B)$,
        for every $i \in [k-1]$,
        since $B-1$ out of the $2B - 1$ free positions are occupied by vertices of $R^i$.
    \end{itemize}
    Set $\mathcal{S}_i = \setdef{a \in \bigcup_{j \in [n]} A^{j,L}}{\pi (\ell^{i-1}_1) < \pi (a) < \pi (\ell^i_1)}$,
    while $\mathcal{S}'_i = \setdef{a \in \bigcup_{j \in [n]} A^{j,R}}{\pi (r^{i-1}_B) < \pi (a) < \pi (r^i_B)}$,
    for every $i \in [k]$.
    Notice that for every $a \in \bigcup_{j \in [n]} A^{j,R}$, $\pi(l^k_1) < \pi (a)$;
    if that were not the case, then $\minofset{\pi}{Y^1} - \pi (a) > 6kB + B + 1 - (4kB + 1) = 2kB + B > b$,
    while $\braces{a, y_1} \in E(G)$ and $y_1 \in Y^1$, which is a contradiction.
    Consequently, sets $\mathcal{S}_i$ partition $\bigcup_{j=1}^n A^{j,L}$, while sets $\mathcal{S}'_i$ partition $\bigcup_{j=1}^n A^{j,R}$.
    Next, define $\mathcal{I}_i = \setdef{j \in [n]}{\mathcal{S}_i \cap A^{j,L} \neq \varnothing}$
    to be the set of indices of the items of the {\UBP} instance for which vertices of the corresponding
    item clique are present in $\mathcal{S}_i$.
    
    \begin{claim}\label{claim:structure3}
        For any $i \in [k]$, it holds that $\mathcal{S}_i = \bigcup_{j \in \mathcal{I}_i} A^{j,L}$
        as well as $\mathcal{S}'_i = \bigcup_{j \in \mathcal{I}_i} A^{j,R}$.
    \end{claim}

    \begin{claimproof}
        We will prove the statement by induction on $i \in [k]$.
        Let $i \in [k]$ and suppose that the statement holds for all $z \in [1, i-1]$ (notice that this is trivially true when $i = 1$).
        We will show that the statement also holds for $i$.

        First, it holds that $\mathcal{I}_i \cap \mathcal{I}_z = \varnothing$ for all $z \in [1, i-1]$,
        as for any such $z$, if $j \in \mathcal{I}_z$ then $A^{j,L} \subseteq \mathcal{S}_z$ due to the induction hypothesis,
        therefore $A^{j,L} \cap \mathcal{S}_i = \varnothing$,
        which implies that $j \notin \mathcal{I}_i$.

        We now argue that $\bigcup_{j \in \mathcal{I}_i} A^{j,R} \subseteq \mathcal{S}'_i$.
        Fix an arbitrary $j \in \mathcal{I}_i$.
        In that case, there exists $a \in A^{j,L} \cap \mathcal{S}_i$ with $\pi (a) < \pi (\ell^i_1)$.
        Additionally, for all $z \in [1, i-1]$ it holds that $j \notin \mathcal{I}_z$, since $\mathcal{I}_i \cap \mathcal{I}_z = \varnothing$,
        thus $A^{j,R} \cap \mathcal{S}'_z = \varnothing$ due to the induction hypothesis.
        Moreover, recall that for every $a' \in A^{j,R}$ it holds that $\pi (a') > \pi (\ell^k_1)$,
        thus $\pi (a') > \pi (r^{i-1}_B)$ follows.
        Assume there exists $a' \in A^{j,R}$ such that $\pi (a') > \pi (r^i_B)$.
        Then, $\pi (a') - \pi (a) > \pi (r^i_B) - \pi (\ell^i_1) = 2kB + B - 1 = b$,
        which is a contradiction since $e = \braces{a, a'} \in E(G)$,
        thus $\stretchEdge{\pi}(e) \leq b$.
        Consequently, it follows that $A^{j,R} \subseteq \mathcal{S}'_i$ for any $j \in \mathcal{I}_i$,
        thus $\bigcup_{j \in \mathcal{I}_i} A^{j,R} \subseteq \mathcal{S}'_i$,
        which implies that $| \bigcup_{j \in \mathcal{I}_i} A^{j,L} | = | \bigcup_{j \in \mathcal{I}_i} A^{j,R} | \leq |\mathcal{S}'_i| = B$.

        Notice that, by the definition of $\mathcal{I}_i$, it holds that $\bigcup_{j \in \mathcal{I}_i} A^{j,L} \supseteq \mathcal{S}_i$,
        which implies that $| \bigcup_{j \in \mathcal{I}_i} A^{j,R} | = | \bigcup_{j \in \mathcal{I}_i} A^{j,L} | \geq |\mathcal{S}_i| = B$.

        Since we have shown both that $\mathcal{S}_i \subseteq \bigcup_{j \in \mathcal{I}_i} A^{j,L}$ while $|\mathcal{S}_i| = |\bigcup_{j \in \mathcal{I}_i} A^{j,L}|$,
        and that $\mathcal{S}'_i \supseteq \bigcup_{j \in \mathcal{I}_i} A^{j,R}$ while $|\mathcal{S}'_i| = |\bigcup_{j \in \mathcal{I}_i} A^{j,R}|$,
        the claim follows.
    \end{claimproof}

    We argue that $(\mathcal{I}_1, \ldots, \mathcal{I}_k)$ certifies that $(A,k)$ is a YES-instance of \UBP.
    \cref{claim:structure3} implies that $\mathcal{I}_{i} \cap \mathcal{I}_{j} = \varnothing$ for all $1 \leq i < j \leq k$.
    Additionally, for every $j \in [n]$, there exists $i \in [k]$ such that $j \in \mathcal{I}_i$,
    since, as already argued, sets $\mathcal{S}_1, \ldots, \mathcal{S}_k$ partition set $\bigcup_{j=1}^n A^{j,L}$.
    Consequently, $(\mathcal{I}_1, \ldots, \mathcal{I}_k)$ partitions $[n]$.
    It remains to show that $\sum_{j \in \mathcal{I}_i} a_j = B$ for all $i \in [k]$.
    In order to do so, notice that $a_j = |A^{j,L}|$, for all $j \in [n]$.
    Moreover, $|\mathcal{S}_i| = B$, while due to \cref{claim:structure3}
    it holds that $\mathcal{S}_i = \bigcup_{j \in \mathcal{I}_i} A^{j,L}$ for all $i \in [k]$.
    Consequently, 
    \[
        \sum_{j \in \mathcal{I}_i} a_j =
        \sum_{j \in \mathcal{I}_i} | A^{j,L} | =
        \biggl| \bigcup_{j\in \mathcal{I}_i} A^{j,L} \biggr| =
        |\mathcal{S}_i|=B,
    \]
    and the statement follows.
\end{proof}

\begin{lemmarep}
    It holds that $\clusterDel(G) = \bO (k)$.
\end{lemmarep}

\begin{proof}
    Let $S = \{ x_{2kB+B} \} \cup \setdef{x_{2iB+B+1}}{i \in [k-1]} \cup \{ y_1 \} \cup \setdef{y_{2iB-B}}{i \in [2,k]}
    \cup \setdef{\ell^i_1, r^i_B}{i \in [k]}$.
    It holds that $|S| = 4k = \bO(k)$, while $G-S$ is a collection of cliques,
    therefore $\clusterDel(G) = \bO(k)$ follows.
\end{proof}

\section{Conclusion}\label{sec:conclusion}

In the current work we extend our understanding of {\BW} in the setting of parameterized complexity.
In particular, we have shown that the problem is FPT when parameterized by the cluster vertex deletion number $\clusterDel$ plus the clique number $\omega$ of the input graph,
becomes W[1]-hard when parameterized solely by $\clusterDel$.

The most natural research direction would be to explore the tractability of the problem when parameterized by twin cover,
modular-width, or vertex integrity,
given the lack of any relevant FPT/XP algorithms or hardness results.
As a matter of fact, it is not even known whether the problem is in XP when parameterized by $\clusterDel$ or tree-depth.

Finally, most tractability results for the various structural parameters rely on some ILP formulation.
This raises the question of whether any other kind of approach is applicable,
as is the case for \CW~\cite{algorithmica/CyganLPPS14}.



\bibliography{bibliography}


\end{document}